\documentclass[journal]{IEEEtran}

\newfont{\bb}{msbm10 scaled 1100}

\newcommand{\RR}{\mbox{\bb R}}
\usepackage{graphicx}
\usepackage{psfrag}
\usepackage{subfigure}
\usepackage{algorithm,algorithmic}
\usepackage{url}
\usepackage{amsmath}
\usepackage{amsfonts,amssymb,amsmath}
\usepackage{mathrsfs}
\usepackage{array}
\usepackage{soul}
\usepackage{color}

\newtheorem{lemma}{Lemma}

\newtheorem{proposition}{Proposition}

\begin{document}
\title{Power Allocation in Two-Hop Amplify-and-Forward MIMO Relay Systems with QoS requirements}
\author{Luca~Sanguinetti, \emph{Member, IEEE}, and~Antonio A.~D'Amico.\thanks{
Copyright \copyright 2012 IEEE. Personal use of this material is permitted. However,
permission to use this material for any other purposes must be obtained from the
IEEE by sending a request to pubs-permissions@ieee.org.
\newline \indent {This research was supported in part by the Seamless Aeronautical Networking through
integration of Data links, Radios, and Antennas (SANDRA) project co-funded by the European Commission within the ``Cooperation Programme'' GA No. FP7- 233679.}
\newline \indent This paper has been submitted in part at the Fourth International Workshop on Computational Advances in Multi-Sensor Adaptive Processing (CAMSAP2011), Puerto Rico, 2011.
\newline \indent The authors are with the
University of Pisa, Department of Information Engineering, Via
Caruso 56126 Pisa, Italy (e-mail: \{luca.sanguinetti, a.damico\}@iet.unipi.it).}}
%

\maketitle

\begin{abstract}
The problem of minimizing the total power consumption while satisfying different quality-of-service (QoS) requirements in a two-hop multiple-input multiple-output network with a single non-regenerative relay is considered. As shown by Y. Rong in \cite{Rong2011}, the optimal processing matrices for both linear and non-linear transceiver architectures lead to the diagonalization of the source-relay-destination channel so that the power minimization problem reduces to properly allocating the available power over the established links. Unfortunately, finding the solution of this problem is numerically difficult as it is not in a convex form. To overcome this difficulty, existing solutions rely on the computation of upper- and lower-bounds that are hard to obtain or require the relaxation of the QoS constraints. In this work, a novel approach is devised for both linear and non-linear transceiver architectures, which allows to closely approximate the solutions of the non-convex power allocation problems with those of convex ones easy to compute in closed-form by means of multi-step procedures of reduced complexity. Computer simulations are used to assess the performance of the proposed approach and to make comparisons with alternatives.
\end{abstract}
\begin{keywords}
MIMO, non-regenerative relay, power allocation, majorization theory, transceiver optimization, quality-of-service requirements, non-convex optimization, power consumption, closed-form solution, decision-feedback equalizer.
\end{keywords}

\section{Introduction}

\PARstart{T}{he} demand for high-speed and high-quality multimedia services in wireless communication systems has increased significantly over the last years.
This has led to a strong interest in multiple-input multiple-output (MIMO) technologies as they represent a promising solution to improve the system reliability and coverage by means of space-time coding technique \cite{Tarokh1999} and/or to increase the spectral efficiency through spatial multiplexing \cite{Telatar1999} and \cite{Foschini98onlimits}. The main impairment of MIMO systems is represented by the
multi-stream interference (MSI) arising from the simultaneous
transmission of parallel data streams over the same frequency
band. However, if channel knowledge is available at the transmitter and receiver side appropriate linear or non-linear transceiver architectures can be used to mitigate MSI and fully exploit the potential benefits of MIMO technologies \cite{PaulrajBook}. The design of transceiver architectures for MSI mitigation in MIMO systems has received great attention in the last years and several excellent works can be found in the open literature. Good surveys of the results obtained in this area can be found in \cite{PalomarBook}\nocite{Shenouda2008} -- \cite{D'Amico2008} in which the authors develop general frameworks for the optimization of linear and non-linear MIMO architectures. This is achieved by resorting to the notions of additive and multiplicative majorization theory (see \cite{MarshallBook} for a complete reference on majorization). 



In the case of a long transmitter-receiver distance, single or multiple repeater (relay) nodes may be necessary to pass information from the transmitter (source) to the receveir (destination). The re-transmission schemes employed at the relays may operate according to several different protocols. Among them, the decode-and-forward protocol makes use of a regenerative relay to first decode the received signal and then re-encode and forward the original information to the destination node. On the other hand, the amplify-and-forward protocol adopts a non-regenerative relay in which the received signal is first linearly processed and then re-transmitted toward the destination. Clearly, amplify-and-forward is more suited to practical implementation than decode-and-forward since decoding and re-encoding multiple data streams involves higher computational complexity and larger processing latency than simply amplifying and forwarding them. Among the different amplify-and-forward MIMO relay systems, the simple two-hop 
single-relay MIMO model has gained a lot of interest over the last years as it provides a reasonable tradeoff between potential benefits and practical implementation issues. For the above reasons, a two-hop single-relay MIMO network is considered in this work.


The optimization of MIMO non-regenerative relay networks has received
much attention recently. A survey of the results obtained in this area can be found in \cite{SanguinettiJSAC2012} and briefly summarized in the following. One of the first attempt in this direction can be found in \cite{Tang2007} and \cite{Medina2007} in which the relay amplifying matrix is designed so as to maximize the link capacity between source and destination. It turns out that optimizing the relay matrix largely improves the system capacity with respect to alternative solutions based on heuristic arguments. Although the capacity is one of the most important information-theoretic measure, there are many other ways of characterizing the reliability of transmission schemes. For this purpose, several different solutions based on linear processing in the form of decorrelating or minimum mean-square-error schemes as well as on non-linear layered architectures have been recently derived and investigated according to different criteria. A unifying framework for the design of linear transceivers in the presence of a single relay is presented in \cite{Rong2009LinearRelayCommunication} and later extended to the multiple relay case in \cite{Rong2009LinearMultiHop}. The main results of \cite{Rong2009LinearRelayCommunication} and \cite{Rong2009LinearMultiHop} are achieved by application of additive majorization and consist in proving that the joint optimization of source, relays and destination matrices under fixed power constraints diagonalizes the MIMO relay channel as long as the objective function is Schur-concave or Schur-convex. In the latter case, the diagonalizing structure is optimal provided that the transmitted data symbols are properly rotated before channel diagonalization. The results in \cite{Rong2009LinearRelayCommunication} and \cite{Rong2009LinearMultiHop} are of great interest since many different optimization criteria driving the design of communication systems arise in connection with Schur-concave or Schur-convex functions \cite{Palomar03}. All the above results have been recently extended in \cite{Rong2009NonLinearRelayCommunication} to the case in which a decision feedback equalizer (DFE) is used at the destination node. This is achieved by means of the multiplicative majorization theory and the equal-diagonal QR decomposition tool illustrated in \cite{Zhang2005}. In particular, it is found that when the objective function is multiplicatively Schur-convex the optimal design leads to the uniform decomposition of the MIMO relay channel into an arbitrary number of identical parallel single-input single-output relay subchannels. This leads to a transmission scheme characterized by a much lower bit-error-rate (BER) than a system employing linear processing at the destination. On the other hand, if the objective function is multiplicatively Schur-concave the optimum non-linear transceiver architecture reduces to the linear one discussed in \cite{Rong2009LinearRelayCommunication} and \cite{Rong2009LinearMultiHop} for which the channel diagonalizing structure is optimal. 

All the aforementioned works are focused on the minimization or maximization of a global objective function subject to fixed power constraints at source and relay nodes. This may prevent their use in those multimedia applications supporting several
types of services each characterized by a different reliability constraint. A common approach to overcome this problem is to deal with the minimization of the power consumption while meeting the quality-of-service (QoS) requirements for each data stream (see for example \cite{Sampath01generalizedlinear} -- \nocite{PalomarQoS2004}\cite{Jiang2006} and references therein). A solution in this direction is proposed in \cite{Guan2008QoSConstraints} in which source and relay matrices are designed so as to ensure a specific signal-to-noise ratio (SNR) on each subchannel (see also \cite{Sayed2007QoSConstraints} for single-antenna relay networks). In \cite{Rong2011}, the author makes use of majorization theory and propose a unifying framework for the design of linear and non-linear transceiver architectures that minimize the total power consumption either in single-hop or multi-hop MIMO relay systems. 
As in \cite{Rong2009LinearRelayCommunication} \nocite{Rong2009LinearMultiHop} -- \cite{Rong2009NonLinearRelayCommunication}, it turns out that the optimal solution leads to the diagonalization of the source-relay-destination channel. However, the resulting power allocation problem is only upper- and lower-bounded using a successive geometric programming approach and a dual decomposition technique, respectively. Unfortunately, the computation complexity of both solutions is relatively high so as to make them unsuited for practical implementation. An alternative solution with reduced complexity is proposed in \cite{Mohammadi2010} where the authors rely on a convex relaxation of the QoS constraints. This results into two convex suboptimal problems that provide different upper and lower bounds to the optimal solution. Unfortunately, such bounds do not meet each other over the entire region of interest but only for high values of SNR. This means that they cannot be used to exactly characterize the solution of the original problem. Furthermore, it is worth observing that relaxing the QoS constraints may result into a suboptimal solution which does not necessarily belong to the feasible set of the original problem. This is the case of the lower bound illustrated in \cite{Mohammadi2010} in which a rescaling operation is required to meet the original QoS constraints, thereby leading to an increase of the power consumption. 


In this work, we return to the problem of designing optimal linear and non-linear transceiver architectures in a two-hop single-relay MIMO network with QoS requirements and we make use of the theoretical analysis presented in \cite{Rong2011} to focus only on the resulting non-convex power allocation problems. The latter are tackled following a different approach that provides a framework in which the power allocation for linear and non-linear schemes can be considered in a unified and concise way, and allows us to extend the results given in \cite{Rong2011}. In particular, the non-convex power allocation problem is first reduced to an equivalent form and then approximated with a convex one, whose solution is within the same feasible set of the original problem and can be computed in closed-form through a multi-step procedure that requires no more than $K$ steps (where $K$ is the number of data streams). Numerical results are used to highlight the effectiveness of the proposed approach. Interestingly, it turns out that for both linear and non-linear systems the approximated solutions are very close to those of the original problems. 



The remainder of this paper is organized as follows\footnote{The following notation is used throughout the paper. Boldface upper and lower-case letters denote
matrices and vectors, respectively, while lower-case letters denote
scalars. We use $\mathbf{A}=\mathrm{diag}\{ a_n\,;\,\,n = 1,2,
\ldots ,K\}$ to indicate a $K \times K$ diagonal matrix with
entries $a_n$ while ${\bf{A}}^{ - 1}$ and
${\bf{A}}^{ 1/2}$ denote the inverse and square-root of a
matrix ${\bf{A}}$. We use ${\bf{I}}_K$ to denote the identity matrix of order $K$ and ${\rm{rank}}({\bf{A}})$ to indicate the rank of a matrix ${\bf{A}}$ while $\left[
\cdot \right]_{k,\ell}$ is the ($k ,\ell$)th entry of the
enclosed matrix. In addition, we use $ {\rm{E}}\left\{ \cdot \right\}$ for
expectation, $\left\| \cdot \right\|$ for the Euclidean norm of
the enclosed vector and the superscript $ ^T$ and $ ^H$ for transposition and Hermitian transposition, respectively. The
notation $0 \le x \bot y \ge 0$ stands for $x y= 0$, $x \ge 0$ and $y \ge 0$. If the elements of ${\bf{x}}$ and ${\bf{y}}$ are arranged in increasing order, we use ${\bf{x}}\prec^{+(w)} {\bf{y}}$ to say that ${\bf{x}}$ is weakly additively majorized
by ${\bf{y}}$.
}. 
Next section describes the two-hop system model and introduces the power minimization problem. In Section III, a linear transceiver architecture is considered and the proposed approximated solution is described together with the convexity analysis of the original power allocation problem. In Section IV, the results are extended to a non-linear architecture in which a DFE is used at the receiver. Simulation results are discussed in Section V while some conclusions are drawn in Section VI.

\section{System description and problem formulation}

We consider a flat-fading\footnote{Although specific for a flat-fading channel, the model adopted throughout
the paper can easily be extended to frequency selective environments
using orthogonal frequency-division multiplexing (OFDM) as a transmission
technique.} MIMO network in
which the information data are carried from source to destination with the aid of a single non-regenerative relay. The information first flows from source to relay and then from relay to destination. The direct link between source and destination is not considered as it is assumed to undergo relatively large attenuation compared to the link via the relay. 

The $k$th symbol is denoted by $s_k$ and is taken from an $L-$ary quadrature amplitude modulation (QAM) constellation with average power normalized to unity for convenience. We denote by $K$ the total number of transmitted symbols and assume that source and destination are equipped with $N$ antennas while the relay has $M$ antennas\footnote{The results can be easily extended to a more general case in which different number of antennas are available at source and
destination.}. 

The source vector ${\bf{s}} = [ {{s}}_1,{{s}}_2, \ldots ,{{s}}_K]^T$ is first linearly processed by an $N \times K$ matrix $\mathbf{U}$ and then transmitted over the source-relay link in the first time-slot. At the relay,
the received signal is processed by an $M \times M$ matrix $\mathbf{F}$ and forwarded to the destination node in the second time-slot. The received signal at the destination takes the form \cite{Rong2009LinearRelayCommunication}
\begin{equation}\nonumber\label{2.1}
{\bf{r}}= \mathbf{H}{\bf{U}}{\bf{s}}+ {\bf{n}}
\end{equation}
where $\mathbf{H} ={\bf{H}}_2{\bf{F}}{\bf{H}}_1$ is the equivalent channel matrix while $\mathbf{H}_{1}$ and $\mathbf{ H}_{2}$ are the source-relay and relay-destination channel matrices, respectively. In addition, ${\bf{n}}$ is a zero-mean Gaussian vector with covariance matrix $\rho\mathbf{R}_\mathbf{n}$, where $\mathbf{R}_\mathbf{n} ={\bf{H}}_2{\bf{F}}{\bf{F}}^H{\bf{H}}_2^{H} + {\bf{I}}_M$ and $\rho>0$ accounting for the noise variance over both links\footnote{The extension to the case in which the noise contribution over each link has a different variance is straightforward.}. 
Henceforth, we denote by 
\begin{equation}\nonumber\label{2.2}
{\mathbf{H}}_1 = {\mathbf{\Omega}_{H_1}\mathbf{\Lambda}_{H_1}^{1/2}\mathbf{V}_{H_1}^H}
\end{equation}
and
\begin{equation}\nonumber\label{2.3}
{\mathbf{H}}_2 = {\mathbf{\Omega}_{H_2}\mathbf{\Lambda}_{H_2}^{1/2}\mathbf{V}_{H_2}^H} 
\end{equation}
the singular value decompositions (SVDs) of ${\mathbf{H}}_1$ and ${\mathbf{H}}_2$.
Without loss of generality, in all subsequent derivations we assume that the entries of the diagonal matrices $\mathbf{\Lambda}_{H_1}$ and $\mathbf{\Lambda}_{H_2}$ are arranged in non-increasing order. This amounts to saying that $\lambda_{H_1,n} \ge \lambda_{H_1,n+1}$ and $\lambda_{H_2,n} \ge \lambda_{H_2,n+1}$ for $n=1,2,\ldots,K-1$.

%
The vector $\bf{r}$ is finally processed by the destination node for data recovery. For this purpose, we consider either a linear receiver or a non-linear detector equipped with a DFE. In both cases, we deal with the joint design of the processing matrices so as to minimize the total power consumption given by 
\begin{equation}\nonumber\label{2.4}
P_T = {\rm{tr}}\left\{ {{\bf{U}} {\bf{U}}^H} \right\} + {\rm{tr}}\left\{ {{\bf{F}} \left({\bf{H}}_1{\bf{U}}{\bf{U}}^H{\bf{H}}_1{^H} + \rho{\bf{I}}_M\right){\bf{F}}^H} \right\}
\end{equation} 
while satisfying different QoS requirements in terms of the MSEs. As mentioned before, the above problem has recently been discussed by Y. Rong in \cite{Rong2011}. In the next, the main results of \cite{Rong2011} are first briefly revised and then the major contributions of this work are described. The linear case is considered first. 

\section{Linear Transceiver Design}
When a linear $K \times N$ receiver $\bf{G}$ is employed, the vector $\bf{y}$ at the input of the decision device takes the form
\begin{equation}
\nonumber
\label{ }
{\bf{y}}={\bf{G}}{\bf{H}}{\bf{U}}{\bf{s}} + {\bf{G}}{\bf{n}}.
\end{equation}
In these circumstances, the constrained power minimization problem can be mathematically formalized as follows \cite{Rong2011}:
\begin{equation}\label{2.5}
\mathop {\min }\limits_{{\bf{U}},{\bf{F}},{\bf{G}}} \;\; P_T \quad {\rm{s}}{\rm{.t}}{\rm{.}}\;\; \left[ {\bf{E}}\right]_{n,n}\le \eta_n\;\;{\rm{for}}\;\; n=1,2,\ldots,K
\end{equation}
where ${\bf{E}}$ denotes the MSE matrix ${\bf{E}} = {\rm{E}}\{
\left( {{\bf{y}} - {\bf{s}}} \right)\left( {{\bf{y}} - {\bf{s}}} \right)^H\}$ while {the quantities $0 < \eta_{n} \le 1$ are design parameters that are assumed to be in non-decreasing order, i.e., $\eta_{n} \le \eta_{n+1}$.}

As proven in \cite{Rong2011}, under the assumption that $K \le \min({\rm{rank}}({\bf{H}}_1),{\rm{rank}}({\bf{H}}_2))$ the optimal $\mathbf{G}^{(o)}$ in \eqref{2.5} is equal to the Wiener filter while the optimal $\mathbf{U}^{(o)}$ and $\mathbf{F}^{(o)}$ have the following form
\begin{equation}\label{3.1}
{\mathbf{U}}^{(o)} ={\tilde{\mathbf{V}}_{H_1}}{\mathbf{\Lambda}}_{U}^{1/2}\mathbf{Q}^H \quad {\text{and}} \quad {\mathbf{F}}^{(o)}= \tilde{\mathbf{V}}_{H_2} {\mathbf{\Lambda}}_F^{1/2}\tilde{\mathbf{\Omega}}_{H_1}^H
\end{equation}
where ${\tilde{\mathbf{V}}_{H_1}}$, $\tilde{\mathbf{V}}_{H_2}$ and $\tilde{\mathbf{\Omega}}_{H_1}$ correspond to the $K$ columns of ${{\mathbf{V}}_{H_1}}$, ${\mathbf{V}}_{H_2}$ and ${\mathbf{\Omega}}_{H_1}$ associated to
the $K$ largest singular values of the corresponding channel matrix while $\mathbf{Q}$ is a suitable $K \times K$ unitary matrix such that 
\begin{equation}\label{3.3}
\left[ {\bf{E}}\right]_{n,n} = \eta_n \;\;{\rm{for}}\;\; n = 1, 2, \ldots,K.
\end{equation}
The matrices $\mathbf{\Lambda}_{U}  = {\rm{diag}}\{\lambda^{(o)}_{U,1}, \lambda^{(o)}_{U,2}, \ldots, $ $\lambda^{(o)}_{U,K}\}$ and $\mathbf{\Lambda}_{F}  = {\rm{diag}}\{\lambda^{(o)}_{F,1}, \lambda^{(o)}_{F,2}, $ $\ldots,\lambda^{(o)}_{F,K}\}$ have diagonal structures with entries 
given by 
\begin{equation} \label{3.4}
 \lambda^{(o)}_{U,n}  = \frac{\rho}{\lambda_{H_1,n}} A^{(o)}_{n} \quad {\text{and}} \quad \lambda^{(o)}_{F,n} = \frac{1}{\lambda_{H_2,n}}  \frac{B^{(o)}_n}{A^{(o)}_n+1}
\end{equation}
where $A^{(o)}_{n}$ and $B^{(o)}_{n}$ are the solutions of the following problem:
\begin{equation}\label{3.6}
\mathop {\min }\limits_{\{A_{n} \, \ge \, 0\}, \{B_{n} \, \ge \, 0\}} \;\;\sum\limits_{n=1}^K \rho \left(\frac{A_{n}}{\lambda_{H_1,n}} + \frac{B_{n}}{\lambda_{H_2,n}}\right)\quad\quad\quad\quad\quad\quad\quad\quad\quad\quad\;
\end{equation}
\begin{equation}\nonumber
\begin{array}{l}
\quad\quad{\rm{s}}{\rm{.t}}{\rm{.}}\quad  \sum \nolimits_{n=1}^{j} {\lambda_{n}}\le\sum \nolimits_{n=1}^{j}\eta_n\;\; {\rm{for}} \;\; j=1,2,\ldots,K\quad\quad\quad\quad\quad
\\ 
\quad\quad \quad\quad\;\;{0 < \lambda_{n} \le 1} \quad n =1,2,\ldots,K \\ 
\;\;\quad\quad \quad\quad{\lambda_{n} \le \lambda_{n+1}}\quad n =1,2,\ldots,K-1.
 \end{array}
\end{equation}
Here, ${\lambda_{n}}$ denotes the $n$th eigenvalue of ${\bf{E}}$ and is given by \cite{Rong2011}
\begin{equation}\label{3.7}
\lambda_{n} = \frac{A_{n}+ B_{n} + 1}{A_{n}+ B_{n} + A_{n} B_{n} + 1}.
\end{equation} 
From (\ref{3.1}), it is seen that ${\mathbf{U}}^{(o)}$ and ${\mathbf{F}}^{(o)}$ are obtained as matched filters along the singular vectors of ${\mathbf{H}}_1$ and ${\mathbf{H}}_2$, respectively. 
In addition, the overall channel matrix given by ${\boldsymbol{\mathcal H}} = {\bf{G}}^{(o)}{\bf{H}}_{2}{\bf{F}}^{(o)}{\bf{H}}_{1} {\bf{U}}^{(o)}$ reduces to ${\boldsymbol{\mathcal H}} = \mathbf{Q}\mathbf{\Lambda}_{\mathcal H}\mathbf{Q}^H$ where $\mathbf{\Lambda}_{\mathcal H}$ is diagonal with entries \cite{Rong2011}
\begin{equation}\nonumber
\lambda_{{\mathcal H},n}= \frac{\lambda^{(o)}_{U,n}\lambda_{H_1,n}\lambda^{(o)}_{F,n}\lambda_{H_2,n}}{\lambda^{(o)}_{U,n}\lambda_{H_1,n}\lambda^{(o)}_{F,n}\lambda_{H_2,n} + \rho\left(\lambda^{(o)}_{F,n}\lambda_{H_2,n} + 1\right)}.
\end{equation}
As discussed in \cite{Rong2011}, the above result reveals that the optimal structure of the relay communication system is diagonal up to a unitary matrix $\bf{Q}$ satisfying \eqref{3.3}. If $K$ is a power of two, such a matrix can be chosen equal to the discrete Fourier transform matrix or to a Walsh-Hadamard matrix \cite{Palomar03}. Otherwise, it can be determined through the iterative procedure described in \cite{Viswanath1999}. 

The only problem left is to solve (\ref{3.6}) or, equivalently, to properly allocating the available power on the established links. This is a challenging task since the quantities $\lambda_{n}$ in \eqref{3.7} are not jointly convex in $A_{n}$ and $B_{n}$, thereby resulting into a combinatorial minimization problem for which no practical algorithm is available. {A way out to this problem consists in solving \eqref{3.7} alternately with respect to $A_{n}$ and $B_{n}$ keeping the other fixed. This leads to an iterative optimization procedure that 
if properly initialized monotonically converges to a local 
optimum of \eqref{3.7} since the conditional updates of 
$A_{n}$ and $B_{n}$ may either decrease or maintain (but not increase) the 
objective function. 
 Although conceptually simple, the above approach requires to iteratively solve multiple convex problems and does not guarantee the convergence to the optimum.} An alternative approach is discussed in \cite{Rong2011} in which the optimal solution is upper- and lower-bounded using a geometric programming approach and a dual decomposition technique, respectively. Unfortunately, the computation complexity of both solutions is relatively high so as to make them unsuited for practical implementation. To overcome the above problems, we follow an alternative approach in which the optimization is first carried out over $A_{n}$ and $B_{n}$ for a fixed $\lambda_{n}$ and then over all possible $\lambda_{n}$ within the feasible set of (\ref{3.6}). 
In Appendix A, using standard calculus techniques it is shown that the quantities $A_{n}$ and $B_{n}$ that minimize the function $\rho({A_{n}}/{\lambda_{H_1,n}} + {B_{n}}/{\lambda_{H_2,n}})$ in \eqref{3.6} for a fixed $\lambda_{n}$ have the form
\begin{equation}\label{3.8}
A_{n}=\frac{1-\lambda_{n}}{\lambda_{n}} + \sqrt{\frac{\lambda_{H_1,n}} {\lambda_{H_2,n}}} \frac{\sqrt{1-\lambda_{n}}}{\lambda_{n}}
\end{equation} 
and
\begin{equation}\label{3.9}
B_{n}=\frac{1-\lambda_{n}}{\lambda_{n}} +  \sqrt{\frac{\lambda_{H_2,n}} {\lambda_{H_1,n}}} \frac{\sqrt{1-\lambda_{n}}}{\lambda_{n}}.
\end{equation}
Using \eqref{3.8} and \eqref{3.9} into $\rho({A_{n}}/{\lambda_{H_1,n}} + {B_{n}}/{\lambda_{H_2,n}})$ yields
\begin{equation}\label{3.10}
P_n\left(\lambda_{n}\right) = \frac{\rho}{\sqrt{\lambda_{H_1,n}\lambda_{H_2,n}}}\left(\gamma_{n} \frac{1-\lambda_{n}}{\lambda_{n}} +  2 \frac{\sqrt{1-\lambda_{n}}}{\lambda_{n}}\right) 
\end{equation} 
where $\gamma_n$ is given by
\begin{equation}\label{3.11}
\gamma_n = \frac{\lambda_{H_1,n} + \lambda_{H_2,n}}{\sqrt{\lambda_{H_1,n}\lambda_{H_2,n}}}.
\end{equation} 
It is worth observing that $\gamma_n$ is a positive function of $\lambda_{H_1,n}$ and $\lambda_{H_2,n}$ taking values in the interval $[2,\infty)$. The minimum $\gamma_{n} = 2$ is achieved for $\lambda_{H_1,n} = \lambda_{H_2,n}$.

The optimization over all possible $\lambda_{n}$ satisfying the constraints in \eqref{3.6} leads to the following equivalent problem:
\begin{equation}\label{3.12}
\mathop {\min }\limits_{\boldsymbol{\lambda} \, \in \, \mathcal L} \;\;P({\boldsymbol{\lambda}}) = \sum\limits_{n=1}^K P_n\left(\lambda_{n}\right) \end{equation}
where $\boldsymbol{\lambda} = [\lambda_1,\lambda_2,\ldots,\lambda_K]^T$ and $\mathcal L$ is the set of admissible $\boldsymbol{\lambda}$ defined as
\begin{equation}\nonumber
\mathcal L = \left\{\boldsymbol{\lambda} \in \RR^K: \boldsymbol{\eta} \prec^{+(w)} \boldsymbol{\lambda} , 0 < \lambda_{n} \le 1 \, \text{and}\, {\lambda_{n} \le \lambda_{n+1}}\right\}
\end{equation}
with $\boldsymbol{\eta} = [\eta_1,\eta_2,\ldots,\eta_K]^T$ being the vector collecting the QoS requirements.
%
%
%
%

To proceed further, we denote by $\boldsymbol{\lambda}^{(o)}=\{\lambda^{(o)}_{1},\lambda^{(o)}_{2},\ldots,\lambda^{(o)}_{K}\}$ the solution of \eqref{3.12} and call $P(\boldsymbol{\lambda}^{(o)})$ the corresponding power consumption. Once computed, the $n$th element of $\boldsymbol{\lambda}^{(o)}$ is used in \eqref{3.8} and \eqref{3.9} for the computation of $A^{(o)}_{n}$ and $B^{(o)}_{n}$, which are then employed in \eqref{3.4} for determining $\lambda^{(o)}_{U,n}$ and $ \lambda^{(o)}_{F,n}$. 

Unfortunately, finding $\boldsymbol{\lambda}^{(o)}$ is a challenging task since the optimization problem in \eqref{3.12} is not in a convex form \cite{BoydBook}. While the feasible set $\mathcal L$ is convex, 
the function $P_n(\lambda_{n}) $ is convex for $0 < \lambda_{n} \le \alpha_{n}$ and concave for $\alpha_{n} < \lambda_{n} \le 1$ where, as shown in Appendix B, $\alpha_{n}$ is given by (see Fig. 1 for a graphical illustration)
\begin{equation}
\label{3.13}
\alpha_{n}=\frac{1}{\dfrac{3}{4}-\dfrac{\sqrt{\xi_n}}{2}+\dfrac{1}{2}\sqrt{\dfrac{3}{4}-\xi_n+\dfrac{1}{\sqrt{\xi_n}} \dfrac{\gamma_n^{2}}{4(\gamma_n^{2}-4)}}}
\end{equation}
with $\xi_n$ defined as follows
\begin{equation}
\label{3.14}
\xi_n=\frac{1}{4}-\dfrac{1}{\sqrt[3]{16(\gamma_n^{2}-4)}}+\dfrac{1}{\sqrt[3]{4(\gamma_n^{2}-4)^{2}}}.
\end{equation}
The two above expressions hold for $\gamma_n \ne 2$. For $\gamma_n = 2$, $\alpha_{n}$ takes its minimum value given by (see Appendix B) 
\begin{equation}\nonumber
\alpha_{{\rm{min}}} = \frac{8}{9}.
\end{equation}
%
%
This result has an interesting theoretical relevance as it allows to prove the following lemma.

\begin{lemma}
If
\begin{equation}\label{3.16}
\sum \limits_{n=1}^{K} {\eta_{n}}\le \frac{8}{9} 
\end{equation}
then the optimization problem in \eqref{3.12} is convex.
\end{lemma}

\vspace{0.1cm}
\begin{proof}
If \eqref{3.16} holds true, then $\sum \nolimits_{n=1}^{K} {\lambda_{n}}\le 8/9$ for all $\boldsymbol{\lambda} \in \mathcal L$
from which (bearing in mind that $\lambda_{n} > 0$) it follows that any admissible $\lambda_{n}$ must be smaller than or equal to $8/9$. Using this fact and recalling that $P_n(\lambda_{n})$ is convex in $(0, 8/9]$, we have that the objective function $P(\boldsymbol{\lambda})$ is convex $\forall \boldsymbol{\lambda} \in \mathcal L$ as it is the sum of convex functions.
\end{proof}

\vspace{0.1cm}
Lemma 1 establishes a sufficient condition for the optimization problem \eqref{3.12} to be convex. Clearly, such a condition is not always met as it depends on the number of data streams and on the specific QoS requirements. This means that solving \eqref{3.12} is in general hard and prompts us to search for alternative methods. As a major contribution of this work, in the next section \eqref{3.12} is approximated with a convex problem whose solution is close to the optimal one and can be
evaluated in closed-form through an exact procedure requiring a maximum number of $K-1$
steps. 

\begin{figure}[t]
\begin{center}
\includegraphics[width=.45\textwidth]{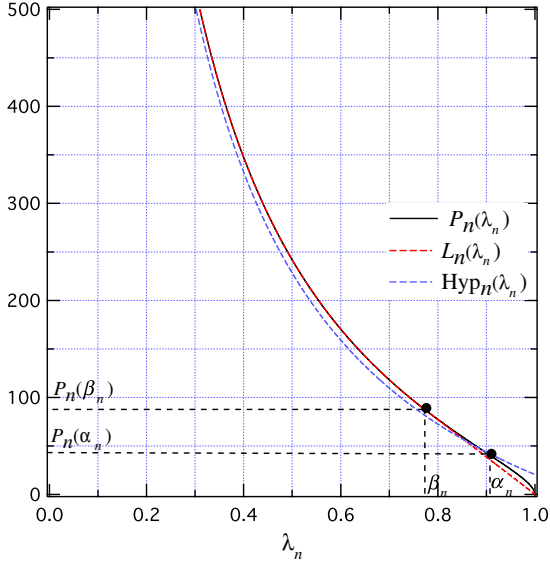}
\end{center}
\caption{Graphical illustration of $P_n(\lambda_n)$, $\text{Hyp}_n\left(\lambda_{n}\right)$ and $L_n(\lambda_n)$ for $0 < \lambda_n \le 1$.}
\label{Fig1}
\end{figure}

\subsection{Hyperbola-based approximation}
We start observing that $P_n(\lambda_{n})$ in \eqref{3.10} has approximately a $1/\lambda_{n}-$shape over the interval $(0,\alpha_{n}]$ in which it is convex. This holds true especially when $\lambda_{n}$ approaches zero. Then, the idea is to approximate $P_n\left(\lambda_{n}\right)$ with the following hyperbola (see Fig. 1):
\begin{equation}\nonumber\label{3.40}
\text{Hyp}_n\left(\lambda_{n}\right) = \frac{w_n}{\lambda_{n}} + z_n
\end{equation} 
where $w_n$ and $z_n$ can be obtained as follows
 \begin{equation}\label{3.41}
(w_n,z_n)=\arg\mathop {\min }\limits_{\tilde w, \tilde z} \mathop {\sup }\limits_{\lambda_{n} \in (0,1]}  \left|P_n\left(\lambda_{n}\right) - \frac{\tilde w}{\lambda_{n}} -  \tilde z\right|.  
\end{equation} 
It can be shown that $w_n$ and $z_n$ are given by\footnote{The proof has been omitted for space limitations. It will be provided upon request. An intuitive explanation at least for $w_n$ relies on the observation that $P_n(\lambda_n)$ in \eqref{3.10} goes to infinity as ${\rho}\left(\gamma_n + 2\right)/(\lambda_n {\sqrt{\lambda_{H_1,n}\lambda_{H_2,n}}})$ for $\lambda_n \rightarrow 0$. This means that when $\lambda_n$ approaches zero the supremum in \eqref{3.41} is bounded only if the coefficients $w_n$ are in the form given by \eqref{3.42}.}
 \begin{equation}\label{3.42}
w_n = \frac{\rho}{\sqrt{\lambda_{H_1,n}\lambda_{H_2,n}}}\left(\gamma_n + 2\right) \end{equation} 
and
 \begin{equation}\label{3.43}
z_n = -\frac{\rho}{\sqrt{\lambda_{H_1,n}\lambda_{H_2,n}}}
\frac{2\gamma_n + 3}{2}.
\end{equation} 
%
Replacing $P_n(\lambda_{n})$ with $\text{Hyp}_n\left(\lambda_{n}\right)$ in \eqref{3.12} leads to the following convex optimization problem:
\begin{equation}\label{3.44}
\mathop {\min }\limits_{\boldsymbol{\lambda} \, \in \, \mathcal L} \;\; \sum\limits_{n=1}^K \frac{w_n}{\lambda_{n}}
\end{equation}
where we have omitted the irrelevant terms $\{z_n\}$. 

The above problem is clearly in a convex form as the objective function and the feasible set are both convex. To solve it, we first observe that the the ordering constraint $\lambda_{n} \le \lambda_{n+1}$ in $\mathcal L$ is always satisfied by the solution of \eqref{3.44}. To see how this comes about, recall that we have assumed $\lambda_{H_1,n} \ge \lambda_{H_1,n+1}$ and $\lambda_{H_2,n} \ge \lambda_{H_2,n+1}$. Then, from (\ref{3.42}) it follows that $w_n \le w_{n+1}$. Now, denote by $\boldsymbol{\lambda}^{\star}$ the solution of \eqref{3.44} and assume that it is such that $\lambda^{\star}_{{n}} \ge \lambda^{\star}_{{n}+1}$ for some ${n} \in \{1,2,\ldots,K-1\}$. Then, the vector $\boldsymbol{\lambda}^{\prime}$ obtained from $\boldsymbol{\lambda}^{\star}$ by simply exchanging $\lambda^{\star}_{{n}}$ and $\lambda^{\star}_{{n}+1}$ would still satisfy the convex constraints of \eqref{3.44} but it would result in a lower objective function since $w_{{n}} \le w_{{n}+1}$. This means that $\boldsymbol{\lambda}^{\star}$ must be such that $\lambda^{\star}_{n} \le \lambda^{\star}_{n+1}$, for $n=1,2,\ldots, K-1$. 

{The numerical evaluation of the solution of \eqref{3.44} cannot be
performed by means of classical bisection or hypothesis testing methods but it requires the development of specific algorithms exploiting the particular structure of the
problem. A good sample in this direction is represented by the multi-level waterfilling algorithm illustrated in \cite{PalomarQoS2004} (see also \cite{Gao2008}) which provides the solution of \eqref{3.44}
in no more than $K(K+1)/2$ iterations. In this work, we
propose an alternative algorithm that allows the computation of $\boldsymbol{\lambda}^{\star}$ in closed-form with a maximum number of $K$ steps.}
\begin{proposition} 
The solution to \eqref{3.44} can be computed through the iterative procedure illustrated in {\bf{Algorithm 1}}.
\end{proposition} 

\vspace{0.1cm}
\begin{proof}
See Appendix C.
\end{proof}

%
%
%
%
%
%
%

\begin{algorithm}[t]
\caption{Multi-step procedure for solving \eqref{3.44}.}

1) Set $i=0$, $\delta_0 = 0$ and compute 
\begin{equation}\nonumber\label{3.70}
\delta_j = \sum\limits_{n=1}^{j}\eta_n
\end{equation}
for $j = 1,2,\ldots,K$.

2) {{\bf{While}}} $i < K$: Compute
\begin{equation}\nonumber\label{3.70.1000}
\lambda_{K-i}^\star = \min(1,\hat \lambda_{K-i})
\end{equation}
where
\begin{equation}\nonumber\label{3.70}
\hat \lambda_{K-i}= \left(\mathop {\max }\limits_{\ell=0,1,\ldots,K-i-1}\;\;\frac{\delta_{K-i} - \delta_\ell}{\sum\nolimits_{j=\ell+1}^{K-i}\sqrt{w_j}} \right) \sqrt{w_{K-i}}.
\end{equation}
Set 
\begin{equation}\nonumber
\delta_j = \left\{ {\begin{array}{*{20}c}
   \delta_j &  \;\;\;{\rm{for}}\;j=0,1,\ldots,K-i-2  \\
\delta_{K-i} - \lambda_{K-i}^\star &  \!\!\!\!\!\!\!\!\!\!\!\!\!\!\!\!{\rm{for}}\;j=K-i-1  \\
\end{array}} \right.
\end{equation}
and $i=i+1$.
%
\end{algorithm}
{As seen from {\bf{Algorithm 1}}, differently from \cite{PalomarQoS2004} a closed-form solution depending on the system parameters is provided for each step of the iterative procedure. It is worth observing that a similar algorithm has been recently proposed in \cite{Fu2011} for the optimization of linear two-hop MIMO networks with multiple relays. Although similar, the proposed one has been derived following a different line of reasoning which is not specific for the optimization problem at hand but it can be used to easily accommodate
other optimization problems with similar structures. 
As we will see in the next section, it can be applied to a MIMO relay network in which a DFE is employed at the destination to compute an approximation of the optimal power allocation.}


%

Once $\boldsymbol{\lambda}^{\star}$ is computed through {\bf{Algorithm 1}}, it is then used for approximating $\lambda^{(o)}_{U,n}$ and $\lambda^{(o)}_{F,n}$ in \eqref{3.4} using the same procedure illustrated before for $\boldsymbol{\lambda}^{(o)}$. More precisely, $\lambda^{\star}_{n}$ is first employed in \eqref{3.8} and \eqref{3.9} to obtain $A^{\star}_{n}$ and $B^{\star}_{n}$, which are then used to replace $A^{(o)}_{n}$ and $B^{(o)}_{n}$ in \eqref{3.4}. This yields 
\begin{equation} \nonumber\label{3.35}
\lambda^{\star}_{U,n}  =
   \frac{\rho}{\lambda_{H_1,n}}\left(\frac{1-\lambda^{\star}_{n}}{\lambda^{\star}_{n}} + \sqrt{\frac{\lambda_{H_1,n}} {\lambda_{H_2,n}}} \frac{\sqrt{1-\lambda^{\star}_{n}}}{\lambda^{\star}_{n}}\right)
\end{equation}
and
\begin{equation} \nonumber\label{3.36}
 \lambda^{\star}_{F,n} =
   \frac{1}{\lambda_{H_2,n}}\frac{\sqrt{\frac{\lambda_{H_1,n}} {\lambda_{H_2,n}}}({1-\lambda^{\star}_{n}}) + \sqrt{1-\lambda^{\star}_{n}}}{\sqrt{\frac{\lambda_{H_2,n}} {\lambda_{H_1,n}}} + \sqrt{1-\lambda^{\star}_{n}}}
\end{equation}
while the required power results equal to $P(\boldsymbol{\lambda}^{\star})$.

It is worth observing that replacing $\boldsymbol{\lambda}^{(o)}$ with $\boldsymbol{\lambda}^{\star}$ inevitably increases the power consumption. Then, we have that 
\begin{equation}\nonumber\label{3.24}
P(\boldsymbol{\lambda}^{\star}) \ge P(\boldsymbol{\lambda}^{(o)})
\end{equation}
with $P(\boldsymbol{\lambda}^{(o)})$ being the global-minimum power allocation for satisfying the QoS requirements. The numerical evaluation of the performance loss $P(\boldsymbol{\lambda}^{\star}) - P(\boldsymbol{\lambda}^{(o)})$ would require knowledge of $P(\boldsymbol{\lambda}^{(o)})$, which can be computed only solving \eqref{3.12}. A possible way out to this problem is to make use of a lower bound of $P(\boldsymbol{\lambda}^{(o)})$, say $\underline P$. More precisely, observing that $\underline P \le P(\boldsymbol{\lambda}^{(o)})$ and $P(\boldsymbol{\lambda}^{(o)}) \le P(\boldsymbol{\lambda}^{\star})$ yields
\begin{equation}\nonumber\label{3.24}
P(\boldsymbol{\lambda}^{\star}) - P(\boldsymbol{\lambda}^{(o)}) \le P(\boldsymbol{\lambda}^{\star}) - \underline P.
\end{equation}
The above result indicates that the performance loss is upper-bounded by the difference $P(\boldsymbol{\lambda}^{\star}) - \underline P$. Clearly, if such a difference is relatively small, the power penalty incurred in approximating ${\boldsymbol{\lambda}}^{(o)}$ with $\boldsymbol{\lambda}^{\star}$ can be neglected and $\boldsymbol{\lambda}^{\star}$ well approximates the solution of the original problem in \eqref{3.12}. Motivated by the above fact, we now proceed with the computation of a lower bound for $P(\boldsymbol{\lambda}^{(o)})$. To this end, we start observing that for any $\lambda_{n} \in (0,1]$ ${P}_n\left(\lambda_{n}\right)$ is not smaller than the following convex function:

\begin{equation}\nonumber\label{3.24}
{L}_n\left(\lambda_{n}\right) = \left\{ {\begin{array}{*{20}c}
   P_n\left(\lambda_{n}\right), & 0 <  \lambda_{n} \le \beta_{n} \\ 
    P_n(\beta_{n}) + P_n'(\beta_{n}) \left(\lambda_{n} - \beta_{n} \right), & \beta_{n} < \lambda_{n} \le 1  \\
\end{array}} \right.
\end{equation}
where 
\begin{equation}\label{3.22}
P_n^{'}\left(\beta_{n}\right) = - \frac{\rho}{\beta_{n}^2\sqrt{\lambda_{H_1,n}\lambda_{H_2,n}}} \cdot\left( \gamma_n + \frac{2-\beta_{n}}{\sqrt{1-\beta_{n}}} \right)
\end{equation}
denotes the first derivative of $P_n(\lambda_{n})$ evaluated at $\lambda_{n} = \beta_{n}$. The latter represents the abscissa of the intersection point between $P_n\left(\lambda_{n}\right)$ and the line $P_n'\left(\beta_n\right)(\lambda_n -1)$ passing through $(1,0)$ and tangent to $P_n(\lambda_{n})$ (see Fig. 1). This amounts to saying that  
\begin{equation}\nonumber
P_n( \beta_n) + P'_n(\beta_n)\left(\lambda_n - \beta_n\right) = P_n'(\beta_n)\left(\lambda_n - 1\right)
\end{equation}
or, equivalently, that $P_n'(\beta_n)(\beta_n -1) = P_n(\beta_n)$ from which using \eqref{3.10} we obtain
\begin{equation}\label{3.22.10}
\left(P_n'(\beta_n)\beta_n + \frac{\rho}{\sqrt{\lambda_{H_1,n}\lambda_{H_2,n}}} \gamma_n \right) \sqrt{1-\beta_n}+2=0.
\end{equation}
Substituting \eqref{3.22} into \eqref{3.22.10} yields $3\beta_n-2=\gamma \sqrt{(1-\beta_n)^3}$. 
%
%
It turns out that the solution of this equation in the interval $(0,1]$ is unique and given by
\begin{equation}\label{3.26}
\beta_{n} = 1- \frac{\epsilon_n}{(\epsilon_n+1)^2}
\end{equation}
with 
\begin{equation}\label{3.27}
\epsilon_n = 
\sqrt[3]{{\frac{{\gamma_n ^2  - 2 + \gamma_n \sqrt {{\gamma_n ^2  - 4} } }}{2}}}.
\end{equation}
Now, consider the following optimization problem:
\begin{equation}\label{3.28}
\mathop {\min }\limits_{\boldsymbol{\lambda}\, \in \, \mathcal L} \;\; L({\boldsymbol{\lambda}}) = \sum\limits_{n=1}^K  L_n\left(\lambda_{n}\right)
\end{equation}
obtained by replacing $P_n(\lambda_{n})$ with $L_n(\lambda_{n})$ in \eqref{3.12}. The above problem is in a convex form since both the objective function and the feasible set are convex. Moreover, denoting by $\underline {\boldsymbol{\lambda}}=\{ \underline {\lambda}_{1},\underline {\lambda}_{2},\ldots,\underline {\lambda}_{K}\}$ its solution we may write
\begin{equation}\nonumber
\label{ineqs}
L(\underline {\boldsymbol{\lambda}}) \le L(\boldsymbol{\lambda}^{(o)})
\end{equation}
where we have used the fact that $\boldsymbol{\lambda}^{(o)}$ belong to the same feasible set of \eqref{3.28}. In addition, as a consequence of the following inequality $L({\boldsymbol{\lambda}}) = \sum\nolimits_{n=1}^K  L_n\left(\lambda_{n}\right) \le \sum\nolimits_{n=1}^K {P}_n\left(\lambda_{n}\right) =P({\boldsymbol{\lambda}})$ we have that
\begin{equation}\nonumber
\label{ineqs}
L(\boldsymbol{\lambda}^{(o)}) \le P(\boldsymbol{\lambda}^{(o)}).
\end{equation}
Collecting the above facts togheter yields 
\begin{equation}\nonumber
\label{ineqs}
L(\underline {\boldsymbol{\lambda}}) \le L(\boldsymbol{\lambda}^{(o)}) \le P(\boldsymbol{\lambda}^{(o)})
\end{equation}
from which it is seen that a lower bound for $P(\boldsymbol{\lambda}^{(o)})$ is given by $\underline P=L(\underline {\boldsymbol{\lambda}})$. Numerical results shown later demonstrate that the difference $P(\boldsymbol{\lambda}^{\star}) - \underline P$ is negligible. As discussed before, this makes $\boldsymbol{\lambda}^{\star}$ a very good approximation of ${\boldsymbol{\lambda}}^{(o)}$. 

Following a simple line of reasoning, it can easily be shown that $\underline P$ represents the tightest lower bound of $P(\boldsymbol{\lambda}^{(o)})$ that can be obtained by replacing $P({\boldsymbol{\lambda}})$ in \eqref{3.12} with a convex lower approximation. To see how this comes about, take a look at Fig. 1 and observe that $L_n(\lambda_n)$ represents the convex function not greater than $P_n(\lambda_n)$ for which the difference $P_n(\lambda_n) - L_n(\lambda_n)$ is minimum for any $\lambda_n \in (0,1]$. This makes $L_n(\lambda_n)$ the best convex lower approximation to $P_n(\lambda_n)$ in the interval $(0,1]$. Accordingly, $L({\boldsymbol{\lambda}})$ is the best convex lower approximation to $P({\boldsymbol{\lambda}})$ in $(0,1]^K$. Since $\underline P = L({\boldsymbol{\lambda}})$, the above statement follows easily.


\section{Non-linear Transceiver Design}
When a non-linear receiver with DFE is employed, the vector $\mathbf{y}$ at the input of the decision device (under the assumption of correct previous decisions) can be written as 

\begin{equation}\nonumber\label{1.3}
{\bf{y}} = \left({\bf{G}}\mathbf{H}\mathbf{U} - \mathbf{B} \right){\bf{s}}+  {\bf{G}}{\bf{n}}
\end{equation}
where $\mathbf{B}$ is a strictly upper triangular matrix $\mathbf{B}$ of order $K$. The power minimization problem is formulated as:
\begin{equation}\label{4.1}
\mathop {\min }\limits_{{\bf{U}},{\bf{F}},{\bf{G}},{\bf{B}}} \;\; P_T \quad{\rm{s}}{\rm{.t}}{\rm{.}}\;\; \left[ {\bf{E}}\right]_{n,n}\le \eta_n\;\;{\rm{for}}\;\; n=1,2,\ldots,K
\end{equation}
where $P_T$ is given by \eqref{2.4} while ${\bf{E}}$ denotes the MSE matrix. 

As proven in \cite{Rong2011}, the optimal ${\bf{G}}^{(o)}$ in \eqref{4.1} is the Wiener filter while ${\bf{B}}^{(o)}$ takes the form ${\bf{B}}^{(o)}= {\bf{D}}{\bf{L}}^H - {\bf{I}}_{K}$. The matrix ${\bf{L}}$ is lower triangular and such that 

\begin{equation}\nonumber
{\bf{L}}{\bf{L}}^H = {\mathbf{U}^{(o)^H}\mathbf{H}^{(o)^H}\mathbf{R}_{\bf{n}}^{(o)^{-1}}\mathbf{H}^{(o)}\mathbf{U}^{(o)} +
\rho {\bf{I}}_{K}}
\end{equation}
with $\mathbf{H}^{(o)} = {\bf{H}}_{2}{\bf{F}}^{(o)}{\bf{H}}_{1}$ and $\mathbf{R}_\mathbf{n}^{(o)} ={\bf{H}}_2{\bf{F}}^{(o)}{\bf{F}}^{(o)^H}{\bf{H}}_2^{H} + {\bf{I}}_M$ while ${\bf{D}}$ is diagonal and designed so as to scale to unity the entries $[{\bf{D}}{\bf{L}}^{H}]_{n,n}$ for $n=1,2,\ldots,K$. The optimal ${\mathbf{U}}^{(o)}$ and ${\mathbf{F}}^{(o)}$ take the form 
\begin{equation}\nonumber\label{4.12}
{\mathbf{U}}^{(o)} ={\tilde{\mathbf{V}}_{H_1}}{\mathbf{\Lambda}}_{U}^{1/2}\mathbf{S}^H \quad {\text{and}} \quad {\mathbf{F}}^{(o)}= \tilde{\mathbf{V}}_{H_2} {\mathbf{\Lambda}}_F^{1/2}\tilde{\mathbf{\Omega}}_{H_1}^H
\end{equation}
where $\mathbf{S}$ is a suitable $K \times K$ unitary matrix such that 
\begin{equation}\label{4.13.1}
\left[{\bf{L}}\right]_{n,n}^{-1} = \sqrt{\eta_n}\;\; \mathrm{for} \;n=1,2,\ldots,K.
\end{equation}
The entries of the diagonal matrices $\mathbf{\Lambda}_U$ and $\mathbf{\Lambda}_F$ are still given by \eqref{3.4} with the quantities $A^{(o)}_n$ and $B^{(o)}_n$ now obtained as:
\begin{equation}\label{4.13}
\mathop {\min }\limits_{\{A_{n} \, \ge \, 0\}, \{B_{n} \, \ge \, 0\}} \;\;\sum\limits_{n=1}^K \rho \left(\frac{A_{n}}{\lambda_{H_1,n}} + \frac{B_{n}}{\lambda_{H_2,n}}\right)\quad\quad\quad\quad\quad\quad\quad\quad\quad\quad\;
\end{equation}
\begin{equation}\nonumber
\begin{array}{l}
\quad\quad{\rm{s}}{\rm{.t}}{\rm{.}}\quad  \prod \nolimits_{n=1}^{j} {\lambda_{n}}\le\prod\nolimits_{n=1}^{j}\eta_n\;\; {\rm{for}} \;\; j=1,2,\ldots,K\quad\quad\quad\quad\quad
\\ 
\quad\quad \quad\quad\;\;{0 < \lambda_{n} \le 1} \quad n =1,2,\ldots,K \\ 
\;\;\quad\quad \quad\quad{\lambda_{n} \le \lambda_{n+1}}\quad n =1,2,\ldots,K-1.
 \end{array}
\end{equation}
As for the linear case, the optimal processing matrices lead to a channel-diagonalizing structure provided that the symbols are properly rotated at the source and destination nodes by a unitary matrix $\bf{S}$ chosen such that \eqref{4.13.1} is satisfied. This is achieved through the iterative algorithm illustrated in \cite{Jiang06thegeneralized}.

Following the same procedure as in the linear case, we first minimize over $A_{n}$ and $B_{n}$ for a fixed $\lambda_{n}$ and then over all possible $\lambda_{n}$ within the feasible set of \eqref{4.13}. The first step yields
\begin{equation}\label{4.12.0}
\mathop {\min }\limits_{\{\lambda_{n} \}} \;\;P({\boldsymbol{\lambda}}) = \sum\limits_{n=1}^K P_n\left(\lambda_{n}\right)\quad\quad\quad\quad\quad\quad\quad\quad
\end{equation}
\begin{equation}\nonumber
\begin{array}{l}
\quad\quad{\rm{s}}{\rm{.t}}{\rm{.}}\quad  \prod \nolimits_{n=1}^{j} {\lambda_{n}}\le\prod \nolimits_{n=1}^{j}\eta_n\;\; {\rm{for}} \;\; j=1,2,\ldots,K\quad\quad\quad\quad\quad
\\ 
\quad\quad \quad\quad\;\;{0 < \lambda_{n} \le 1} \quad n =1,2,\ldots,K \\ 
\;\;\quad\quad \quad\quad{\lambda_{n} \le \lambda_{n+1}}\quad n =1,2,\ldots,K-1.
 \end{array}
\end{equation}
Observe that the feasible set of the above problem is not convex since the inequality constraints
\begin{equation}\label{4.12.1}
\prod \nolimits_{n=1}^{j} {\lambda_{n}}\le\prod \nolimits_{n=1}^{j}\eta_n\;\; {\rm{for}} \;\; j=1,2,\ldots,K
\end{equation}
do not form a convex set for $K\ge 2$ \cite{BoydBook}. Letting 
\begin{equation}\nonumber
\theta_n = \ln \lambda_n \quad \text{and}\quad \kappa_n = \ln \eta_n
\end{equation}
the problem can be readily reformulated as follows
\begin{equation}\label{4.12}
\mathop {\min }\limits_{\boldsymbol{\theta} \, \in \, \mathcal N} \;\;Q({\boldsymbol{\theta}}) = \sum\limits_{n=1}^K Q_n\left(\theta_{n}\right) \end{equation}
where $\boldsymbol{\theta} = [\theta_1,\theta_2,\ldots,\theta_K]^T$ and $Q_n\left(\theta_{n}\right)$ takes the form
\begin{equation}\label{func_Q}
Q_n\left(\theta_{n}\right) = \frac{\rho}{\sqrt{\lambda_{H_1,n}\lambda_{H_2,n}}}\left[\gamma_{n} \left(\frac{1}{e^{\theta_n}} - 1\right) +  2 \frac{\sqrt{1- e^{\theta_n}}}{e^{\theta_n}}\right]
\end{equation} 
while $\mathcal N$ is defined as
\begin{equation}\nonumber
\mathcal N = \left\{\boldsymbol{\theta} \in \RR^K: \boldsymbol{\kappa} \prec^{+(w)} \boldsymbol{\theta} , \theta_{n} \le 0\, \text{and}\, {\theta_{n} \le \theta_{n+1}}\right\}
\end{equation}
with $\boldsymbol{\kappa} = [\kappa_1,\kappa_2,\ldots,\kappa_K]^T$. 

Although $\mathcal N$ is convex, the problem \eqref{4.12} is not in a convex form due to the properties of function $Q_n(\theta_n)$. The latter is shown in Appendix D to be 
convex for $\theta_{n} \le \ln \phi_{n}$ and concave for $\ln \phi_{n} < \theta_{n} \le 0$ (see Fig. 2) where $\phi_{n}$ is given by 
\begin{equation}
\label{4.18}
\phi_{n}=1 - \frac{1}{2+\frac{\gamma_n}{\sqrt {\chi_n}} + \sqrt{\gamma_n \sqrt{\chi_n} + 5 - \frac{\gamma_n^2}{\chi_n}}}
\end{equation}
with
\begin{align}\nonumber
& \chi_n=\frac{8}{3} + \frac{\sqrt[3]{2}}{3} \sqrt[3]{27\gamma_n^2 - 104 + \sqrt{(27 \gamma_n^2 - 104)^2-16}}  + \\ & \qquad \qquad + \frac{\sqrt[3]{2}}{3} \sqrt[3]{27\gamma_n^2 - 104 - \sqrt{(27 \gamma_n^2 - 104)^2-16}}.\nonumber
\end{align}
The minimum value of $\phi_{n}$ is achieved for $\gamma_n=2$ and given by (see Appendix D)
\begin{equation}
\label{4.18.1}
\phi_{\min}=2\left(\sqrt{2} - 1\right).
\end{equation}
It is worth observing that the above result cannot be used as done in the previous section to derive a sufficient condition on the QoS requirements under which the convexity of \eqref{4.12} can be established. 

As for the linear case, in the next \eqref{4.12} is replaced with a convex problem easy to solve whose solution is shown to be close to the optimal one.

\begin{figure}[t]
\begin{center}
\includegraphics[width=.45\textwidth]{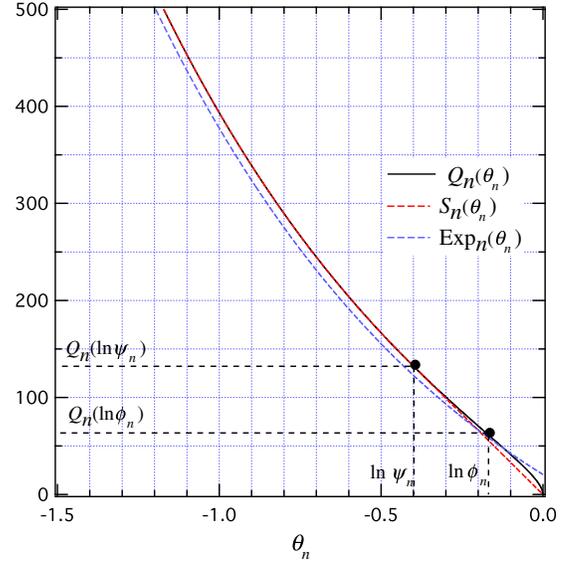}
\end{center}
\caption{Graphical illustration of $Q_n(\theta_{n})$, $\text{Exp}_n(\theta_{n})$ and $S_n(\theta_{n})$ for $-1.5\le \theta_{n} \le 0$.}
\label{Fig1}
\end{figure}

\subsection{Exponential-based approximation}
Using the same methodology illustrated in the previous section, $Q_n(\theta_{n}) $ is replaced by the following exponential function: 

\begin{equation}\nonumber
\text{Exp}_n(\theta_{n}) = \frac{w_n} {e^{\theta_n}} + z_n 
\end{equation}
where $w_n$ and $z_n$ are obtained as follows
 \begin{equation}\nonumber
(w_n,z_n)=\arg\mathop {\min }\limits_{\tilde w, \tilde z} \mathop {\sup }\limits_{\theta_{n} \in  (-\infty,0]}  \left|Q_n\left(\theta_{n}\right) - \frac{\tilde w} {e^{\theta_n}} -  \tilde z\right|.  
\end{equation} 
Using standard analysis not shown for space limitations, it is found that $w_n$ and $z_n$ are still given by \eqref{3.42} and \eqref{3.43}.

Then, the non-convex power allocation problem in \eqref{4.12} is approximated with the following convex one:
\begin{equation}\label{4.44}
\mathop {\min }\limits_{\boldsymbol{\theta} \, \in \, \mathcal N} \;\;\sum\limits_{n=1}^K \frac{w_n} {e^{\theta_n}}
\end{equation}
whose solution $\boldsymbol{\theta}^{\star} = [{\theta}^{\star}_1, {\theta}^{\star}_2,\ldots,{\theta}^{\star}_K]^T$ is easily proven to be such that ${\theta}^{\star}_n \le {\theta}^{\star}_{n+1}$. Interestingly, it turns out that $\boldsymbol{\theta}^{\star}$ can be computed using an iterative procedure derived following the same line of reasoning adopted to obtain {\bf{Algorithm 1}}.

\begin{proposition} 
The solution to \eqref{4.44} can be obtained through the iterative procedure illustrated in {\bf{Algorithm 2}}.
\end{proposition} 
\begin{proof}
See Appendix E.
\end{proof}

%

Once $\boldsymbol{\theta}^{\star}$ is computed through {\bf{Algorithm 2}}, it is then used to approximate $\lambda^{(o)}_{U,n}$ and $\lambda^{(o)}_{F,n}$ exactly in the same way illustrated for the linear case. This produces
\begin{equation} \nonumber
\lambda^{\star}_{U,n}  =
   \frac{\rho}{\lambda_{H_1,n}}\left(\frac{1-e^{{\theta}_n^{\star}}}{e^{{\theta}_n^{\star}}} + \sqrt{\frac{\lambda_{H_1,n}} {\lambda_{H_2,n}}} \frac{\sqrt{1-e^{{\theta}_n^{\star}}}}{e^{{\theta}_n^{\star}}}\right)
\end{equation}
and
\begin{equation} \nonumber
 \lambda^{\star}_{F,n} =
   \frac{1}{\lambda_{H_2,n}}\frac{\sqrt{\frac{\lambda_{H_1,n}} {\lambda_{H_2,n}}}({1-e^{{\theta}_n^{\star}}}) + \sqrt{1-e^{{\theta}_n^{\star}}}}{\sqrt{\frac{\lambda_{H_2,n}} {\lambda_{H_1,n}}} + \sqrt{1-e^{{\theta}_n^{\star}}}}
\end{equation}
where we have used the fact that $\lambda_n^\star =e^{{\theta}_n^{\star}}$.

%
%
%
%
%
%
%

\begin{algorithm}[t]
\caption{Multi-step procedure for solving \eqref{4.44}.}

1) Set $i=0$, $\delta_0 =0$ and compute 
\begin{equation}\nonumber\label{3.70}
\delta_j = \sum\limits_{n=1}^{j}\kappa_n = \sum\limits_{n=1}^{j}\ln \eta_n
\end{equation}
for $j = 1,2,\ldots,K$.

\vspace{0.1cm}

2) {{\bf{While}}} $i < K$: Compute
\begin{equation}\nonumber\label{3.70.1000}
\theta_{K-i}^\star = \min(1,\hat \theta_{K-i})
\end{equation}
where
\begin{equation}\nonumber\label{3.70}
\hat \theta_{K-i}= \ln w_{K-i} + 
\mathop {\max }\limits_{\ell=0,1,\ldots,K-i-1} \dfrac{ \delta_{K-i} - \delta_{\ell}- \sum\nolimits_{j=\ell+1}^{K-i} \ln w_j}{K-i - \ell}.
\end{equation}
Set 
\begin{equation}\nonumber
\delta_j = \left\{ {\begin{array}{*{20}c}
   \delta_j &  \;\;\;{\rm{for}}\;j=0,1,\ldots,K-i-2  \\
\delta_{K-i} - \theta_{K-i}^\star &  \!\!\!\!\!\!\!\!\!\!\!\!\!\!\!\!{\rm{for}}\;j=K-i-1  \\
\end{array}} \right.
\end{equation}
and $i=i+1$.
%
\end{algorithm}
In order to validate the quality of $\boldsymbol{\theta}^{\star}$, a lower bound is now computed. Paralleling the steps of the previous section, the lower bound is obtained as $\underline Q = S (\underline {\boldsymbol{\theta}})$ where $\underline {\boldsymbol{\theta}}$ is the solution of the following problem:
\begin{equation}\label{4.28}
\mathop {\min }\limits_{\boldsymbol{\theta}\, \in \, \mathcal N} \;\;S({\boldsymbol{\theta}}) = \sum\limits_{n=1}^K S_n\left(\theta_{n}\right).
\end{equation}
As shown in Fig. 2, $S_n\left(\theta_{n}\right)$ is given by
\begin{equation}\nonumber\label{4.29}
\!{S}_n\left(\theta_{n}\right) \!\!= \!\!\left\{ {\begin{array}{*{20}c}
   \!\!\!\!Q_n\left(\theta_{n}\right), & \!\!\!\theta_{n} \le \ln\psi_{n} \\ 
    \!\!\!Q_n(\ln \psi_{n}) \!+ \!Q_n'(\ln \psi_{n}) \!\!\left(\theta_{n} - \ln \psi_{n} \right), &  \!\!\!\ln\psi_{n} < \theta_{n} \le 0   \\
\end{array}} \right.
\end{equation}
where 
\begin{equation}\nonumber
Q_n^{'}\left(\ln \psi_{n}\right) = - \frac{\rho }{\psi_{n}\sqrt{\lambda_{H_1,n}\lambda_{H_2,n}}} \left(\gamma_{n} + \frac{2 - {\psi_{n}}}{\sqrt{1-{\psi_{n}}}}\right)
\end{equation}
is the first derivative of $Q_n(\theta_{n})$ evaluated at $\theta_{n} = \ln \psi_{n}$ which represents the abscissa of the intersection point between $Q_n\left(\theta_{n}\right)$ and the line passing through the origin and tangent to $Q_n(\theta_{n})$. In Appendix D, it shown that $\psi_{n}$ is such that  
\begin{equation}\label{4.32}
\gamma_n\sqrt{1 - \psi_{n}} +  2  +\frac{\ln \psi_{n}}{\sqrt{1 - \psi_{n}}} \left(\gamma_n+\frac{2 - {\psi_{n}}}{\sqrt{1-{\psi_{n}}}}\right)= 0.
\end{equation}
Unfortunately, the solution of the above equation cannot be computed in closed-form but it can only be evaluated numerically.

%
%
%
%
As for the linear case, numerical results shown later demonstrate that the difference between $Q(\boldsymbol{\theta}^{\star})$ and $\underline Q$ is negligible. Moreover, $\underline Q$ still represents the best convex lower approximation of the original problem in \eqref{4.12} .

\section{Numerical results}

Numerical results are now given to highlight the effectiveness of the proposed solutions. Comparisons are made with the successive geometric programming (GP) approach illustrated in \cite{Rong2011}. The CVX convex optimization toolbox for MATLAB is used to solve the optimization problem. The number of antennas employed at the source, relay and destination nodes is equal, i.e., $N=M$. The number of transmitted symbols is fixed to $K=N$. The entries of ${\bf{H}}_1$ and ${\bf{H}}_2$ are modeled as independent complex circularly symmetric Gaussian random variable with zero mean and variance $1/N$. All numerical results are obtained averaging over $10^3$ independent realizations of the channels.

\begin{figure}[t]
\begin{center}
\includegraphics[width=.45\textwidth]{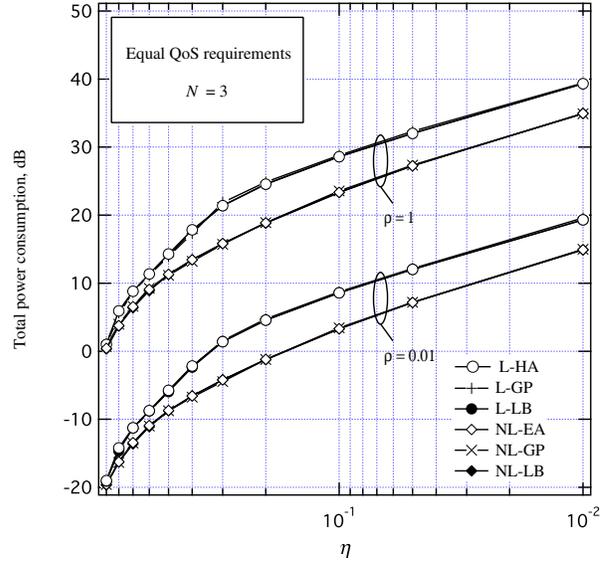}
\end{center}
\caption{Total power consumption when equal QoS constraints are given with $N=3$ and $\rho = 1$ or $0.01$.}
\label{picture1}
\end{figure}

\begin{figure}[t]
\begin{center}
\includegraphics[width=.45\textwidth]{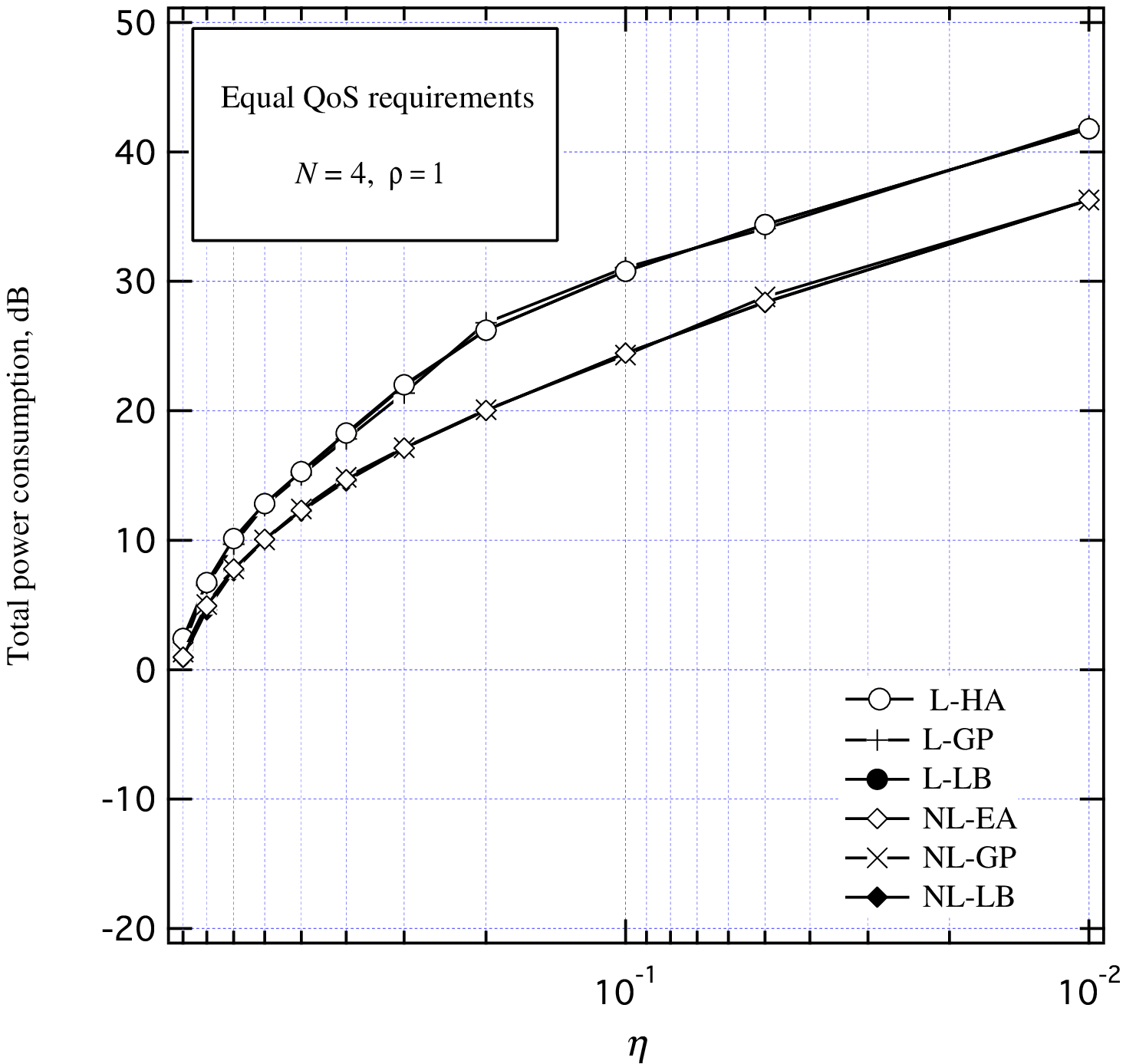}
\end{center}
\caption{Total power consumption when equal QoS constraints are given with $N=4$ and $\rho = 1$.}
\label{picture1}
\end{figure}

\begin{table}[t]
\caption{Total power consumption when equal QoS constraints are given with $N=3$ and $\rho = 1$.}
\label{table01}\renewcommand{\arraystretch}{1.3} \centering
\begin{tabular}{|c||c||c||c||c||c|}
\hline
Algorithm & $\eta=0.9$ & $\eta=0.5$ & $\eta=0.1$ & $\eta=0.05$ & $\eta=0.01$\\ \hline
L-HA & $1.001$ & $14.211$ & $28.5907$ & $32.020$ & $39.316$\\ \hline
L-GP & $1.109$ & $14.254$ & $28.612$ & $32.121$ & $39.321$ \\ \hline
L-LB & $0.990$ & $14.113$ & $28.5907$ & $32.020$ & $39.316$ \\ \hline
NL-EA & $0.356$ & $11.291$ & $23.351$ & $27.248$ & $34.960$ \\ \hline
NL-GP & $0.361$ & $11.314$ & $23.359$ & $27.250$ & $34.981$ \\ \hline
NL-LB & $0.325$ & $11.252$ & $23.330$ & $27.247$ & $34.959$ \\ \hline
\end{tabular}%
\end{table}

\begin{figure}[t]
\begin{center}
\includegraphics[width=.45\textwidth]{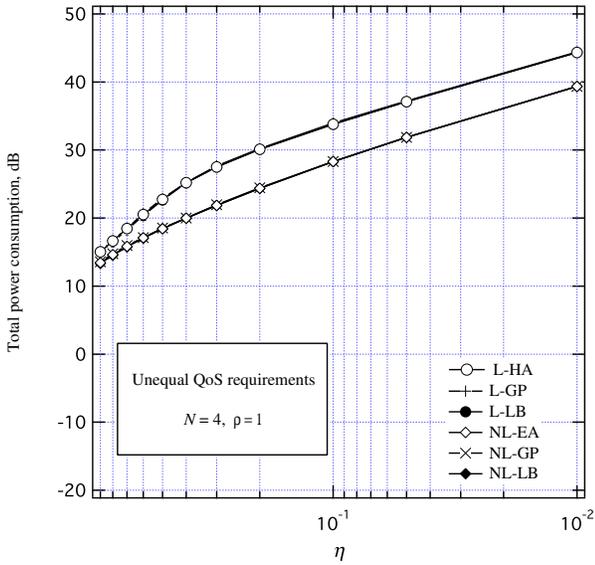}
\end{center}
\caption{Total power consumption when unequal QoS constraints are given with $N=4$ and $\rho = 1$.}
\label{picture1}
\end{figure}

Fig. 3 illustrates the total power consumption as a function of the QoS constraints when $N=3$ and the noise variance is $\rho = 1$ or $0.01$. The same QoS constraint is imposed for each data stream, i.e., $\eta_n = \eta$ for $n=1,2,\ldots,K$. The curves labelled with L-HA and NL-EA refer to a system in which the power is allocated using {\bf{Algorithm 1}} and {\bf{Algorithm 2}}, respectively. On the other hand, L-GP and NL-GP refer to a system in which the successive GP approach of \cite{Rong2011} is employed in conjunction with a linear or a nonlinear receiver, respectively. The results of Fig. 3 indicate that the total power consumption of the proposed solutions is substantially the same as that obtained with the schemes discussed in \cite{Rong2011}. However, this is achieved with much lower complexity as L-HA and NL-EA are obtained in closed form by means of {\bf{Algorithm 1}} and {\bf{2}}. Moreover, both algorithms attain the corresponding lower bounds denoted by L-LB and NL-LB and evaluated solving \eqref{3.28} and \eqref{4.28}. To ease comparisons, some of the results of Fig. 3 are also listed in Table \ref{table01}. As already shown in \cite{Rong2011}, it turns out that the non-linear architecture provides the best performance for all the investigated values of $\eta$. As expected, increasing the noise variance from $0.01$ to $1$ enhances the required power of all the investigated solutions of approximately $20$ dB. Similar conclusions can be drawn from the results of Fig. 4 in which $N=4$ and $\rho = 1$. 

The results of Fig. 5 are obtained in the same operating conditions of Fig. 4 except that now  $\eta_1 = \eta/4$, $\eta_2 = \eta_3 = \eta/2$ and $\eta_4 = \eta$ and $\rho = 1$. Compared to the results of Fig. 4, the total power consumption increases due to the more stringent requirements over some established links.

\section{Conclusions}

We have discussed the power allocation problem in a two-hop MIMO relay network in which the goal is the minimization of the total power consumption while satisfying different QoS requirements given in terms of the MSEs. The original non-convex power allocation problem has been approximated for linear and non-linear architectures with a convex one that can be solved exactly through a multi-step procedure of reduced complexity. Comparisons with existing alternatives requiring much higher computational burden show that the same total power consumption is required. It is worth observing that the extension of the proposed approach to the case in which multi-hops are used to carry the information from the source to the destination is not simple and it is currently under investigation.

\section*{Appendix A}
Keeping $\lambda_{n}$ fixed, from \eqref{3.7} we have that
\begin{equation}\label{A.1}
B_{n}=\frac{A_n}{\lambda_{n} + \lambda_{n}A_n - 1} -1
\end{equation} 
while $\rho ({A_{n}}/{\lambda_{H_1,n}} + {B_{n}}/{\lambda_{H_2,n}})$ reduces to
\begin{align}\nonumber
&\rho \left(\frac{A_{n}}{\lambda_{H_1,n}} + \frac{B_{n}}{\lambda_{H_2,n}}\right) = \\ &\qquad \rho \left[\frac{A_n}{\lambda_{H_1,n}} + \frac{1}{\lambda_{H_2,n}}\left(\frac{A_n}{\lambda_{n} + \lambda_{n}A_n - 1} -1\right)\right].\nonumber
\end{align} 
Taking the derivative with respect to $A_{n}$ and equating it to zero yields
\begin{equation}\nonumber
A_{n}=\frac{1-\lambda_{n}}{\lambda_{n}}
 \pm  \sqrt{\frac{\lambda_{H_1,n}} {\lambda_{H_2,n}}} \frac{\sqrt{1-\lambda_{n}}}{\lambda_{n}}
\end{equation} 
from which using \eqref{A.1} yields
\begin{equation}\nonumber
B_{n}=\frac{1-\lambda_{n}}{\lambda_{n}}
 \pm   \sqrt{\frac{\lambda_{H_2,n}} {\lambda_{H_1,n}}} \frac{\sqrt{1-\lambda_{n}}}{\lambda_{n}}.
\end{equation} 
We now observe that the solution
\begin{equation}\label{A.2}
A_{n}=\frac{1-\lambda_{n}}{\lambda_{n}}
-  \sqrt{\frac{\lambda_{H_1,n}} {\lambda_{H_2,n}}} \frac{\sqrt{1-\lambda_{n}}}{\lambda_{n}} \quad 
\end{equation} 
and
\begin{equation}\label{A.2.1}
B_{n}=\frac{1-\lambda_{n}}{\lambda_{n}}
-  \sqrt{\frac{\lambda_{H_2,n}} {\lambda_{H_1,n}}} \frac{\sqrt{1-\lambda_{n}}}{\lambda_{n}}
\end{equation} 
is not admissible as it violates one of the constraints $A_{n} \ge 0$ and $B_{n} \ge 0$ in (\ref{3.6}). To see how this comes about, note that $A_{n}$ in (\ref{A.2}) is larger than zero only if ${{\lambda_{H_1,n}}} \le {\lambda_{H_2,n}} (1-\lambda_{n})$. On the other hand, from (\ref{A.2.1}) we have that $B_{n}\ge 0$ only if $\lambda_{H_2,n} \le \lambda_{H_2,n} (1-\lambda_{n})$. Collecting these two facts together, it is seen that the constraints $A_{n} \ge 0$ and $B_{n} \ge 0$ are both satisfied when $\lambda_{H_1,n} /(1-\lambda_{n}) \le \lambda_{H_2,n} \le \lambda_{H_1,n} (1-\lambda_{n})$. Since $0 < \lambda_{n} \le 1$ it follows that such inequalities cannot be ensured by any value of $\lambda_{H_1,n}$ and $\lambda_{H_2,n}$. Then, the optimal $A_{n}$ and $B_{n}$ are given by \eqref{3.8} and \eqref{3.9}, respectively.

\section*{Appendix B}


The convexity of the objective function $P_n(\lambda_{n})$ in (\ref{3.12}) on $\mathcal P = \{\lambda_{n}; 0 < \lambda_{n} \le 1\}$ is studied. For notational convenience, the index $n$ is omitted. 

We start taking the derivative of $P\left(\lambda\right)$ in (\ref{3.10}) with respect to $\lambda$ and obtain 
\begin{equation}\label{B.3}
P{'}\left(\lambda\right) =- \frac{\rho}{ \lambda^2\sqrt{\lambda_{H_1}\lambda_{H_2}}} \cdot\left( \gamma + \frac{2-\lambda}{\sqrt{1-\lambda}} \right) 
\end{equation} 
from which it easily follows that $P{'}(\lambda) < 0$ for all $\lambda \in \{\mathcal P\setminus 1\}$. This means that $P\left(\lambda\right)$ is a decreasing function of $\lambda$. Taking the derivative of $P'\left(\lambda\right)$ yields
\begin{equation}\label{B.4}
P{''}\left(\lambda\right) =\frac{\rho}{2\lambda^3\sqrt{\lambda_{H_1}\lambda_{H_2}}} \cdot \frac{f(\lambda) - \gamma g(\lambda)}{\left(1-\lambda\right)\sqrt{1-\lambda}}
\end{equation} 
%
where we have defined $f(\lambda) = 3\lambda^2 - 12\lambda + 8$ and $g( \gamma) = -4\left(1-\lambda\right)\sqrt{1-\lambda}$. Setting $f( \lambda) =  \gamma g(\lambda)$ 
leads to the following quartic equation
\begin{equation}\nonumber
9\lambda^4 - 8\left(9-2\gamma^{2}\right)\lambda^3+\left(4-\gamma^{2}\right)\left(48\lambda^2-48\lambda+16\right)= 0.
\end{equation}
After lengthy computations (not shown for space limitations) it turns out that there exists an \emph{unique} point $\lambda$ in $\mathcal P$ solving the above equation. For $\gamma \ne 2$, such a point $\lambda = \alpha$ is given by \eqref{3.13} in the text. On the other hand, setting $\gamma= 2$ into the quartic equation produces $\lambda^3(9 - 8\lambda)=0$ from which it follows that $\alpha = 8/9$. Once $\alpha$ has been computed through \eqref{3.13} and \eqref{3.14}, we have that $f(\lambda) - \gamma g( \lambda) \ge 0$ for any $0 < \lambda \le \alpha$. This amounts to saying that $P{''}\left(\lambda\right) \ge 0$ or, equivalently, that $P\left(\lambda\right)$ is convex over the convex set $0 < \lambda \le \alpha$. Vice versa, for any $\alpha < \lambda < 1$ we have  $f(\lambda) - \gamma g(\lambda) < 0$ so that $P{''}\left(\lambda\right) < 0$ and the function $P\left(\lambda\right)$ is concave.

A close inspection of \eqref{3.13} and \eqref{3.14} reveals that $\alpha$ depends exclusively on $\gamma$. We are now interested in computing its minimum value $\alpha_{{\rm{min}}}$ as $\gamma$ varies. Although $\alpha_{{\rm{min}}}$ could be in principle computed using the closed-form expression given in \eqref{3.13}, a more simple line of reasoning is followed henceforth. Observe first that $f(\lambda)$ in (\ref{B.4}) is a convex and decreasing function while $ g(\lambda)$ is a negative, concave and increasing function in $\mathcal P$. Moreover, recall that $\gamma$ is a positive function of $\lambda_{H_1}$ and $\lambda_{H_2}$ taking values in the interval $[2,\infty)$. Collecting these facts together, it follows that the point $\alpha$ such that $f(\alpha) = \gamma g(\alpha)$ moves towards the upper limit of $\mathcal P$ as $\gamma$ increases. This means that its maximum value is given by $\alpha_{{\rm{max}}}=1$ and it is achieved when $\gamma$ goes to infinity. On the other hand, $\alpha_{{\rm{min}}}$ is achieved when $\gamma$ takes its minimum value, i.e., $\gamma = 2$. As shown before, this yields $\alpha_{{\rm{min}}} = 8/9$.

\section*{Appendix C}

In the next, we highlight the major steps leading to the solution of (\ref{3.44}) in the form given by {\bf{Algorithm 1}}. For this purpose, we adopt the following approach. We first compute $\hat \lambda_K$, which represents the largest solution of the problem obtained from (\ref{3.44}) after removing the constraints $\lambda_{n} \le 1$ for $n=1,2,\ldots, K$. Interestingly, $\hat \lambda_K$ can be obtained in closed-form and it can be efficiently used to compute $\lambda_K^\star$. All the above results are then used to derive a simple iterative procedure providing all the remaining solutions $\lambda_{n}^\star$ for $n=1,2,\ldots,K-1$ in no more than $K-1$ iterations.

The first step requires to find the largest solution of the following problem:
\begin{equation}\label{C.2}
    \begin{array}{cll}
      \;\quad  \underset{\{0\,<\, \lambda_n\}}{\min} &
\sum\nolimits_{n=1}^{K}\dfrac{w_n}{\lambda_n}
&  \\
      \;\;\quad \mathrm{s.t.} & \sum\nolimits_{n=1}^{j}\lambda_n \le \delta _j & \quad j=1,2,\ldots,K
    \end{array}
\end{equation}
where
\begin{equation}\label{C.1}
\delta _j = \sum \nolimits_{n=1}^{j} \eta_n.
\end{equation}
To this end, we denote by $\mathcal{P}^{(i)}$ for $i=0,1,2,\ldots,K-1$ the optimization problem obtained from (\ref{C.2}) by removing the inequality constraints from the $(i+1)$th to the $(K-1)$th, i.e.,
\begin{equation}\nonumber
    \begin{array}{cll}
      \mathcal{P}^{(i)}: \quad \quad\underset{\{0\,<\, \lambda_n\}}{\min} &
\sum\nolimits_{n=1}^{K}\dfrac{w_n}{\lambda_n}
&  \\
      \;\;\quad\quad\quad\quad\quad \mathrm{s.t.} & \sum\nolimits_{n=1}^{j}\lambda_n \le \delta _j & \quad j=1,2,\ldots,i
      \\
      & \sum\nolimits_{n=1}^{K}\lambda_n \le \delta _K &
    \end{array}
\end{equation}
from which it follows that $\mathcal{P}^{(K-1)}$ is equivalent to (\ref{C.2}) while $\mathcal{P}^{(0)}$ is obtained from (\ref{C.2}) by removing all the constraints except the last. 
Since the above problem satisfies the Slater's condition, the $n$th solution of $\mathcal P^{(i)}$, say $\hat {\lambda}_n^{(i)}$, is found from the necessary and sufficient Karush-Kuhn-Tucker (KKT) optimality conditions. This yields
\begin{equation}\label{C.6}
\; \hat {\lambda}_n^{(i)} =    \sqrt{\frac{{w_n}}{{\sum\nolimits_{j=n}^{i} \varsigma_j^{(i)}+\varsigma_K^{(i)}}}}\quad \; {\rm{for}}\; n=1,2,\ldots,i \quad\quad\;
\end{equation}
and
\begin{equation}\label{C.6.1}
\hat {\lambda}_n^{(i)} =  \sqrt{\frac{{w_n}}{{\varsigma_K^{(i)}}}}\quad \quad {\rm{for}}\; n=i+1,i+2,\ldots,K
\end{equation}
where the Lagrange multipliers  $\varsigma_j^{(i)}$ are chosen to satisfy the following constraints:
\begin{equation}\label{C.7}
0 \le \varsigma_j^{(i)} \bot \left(\delta _j - \sum\nolimits_{n=1}^{j}\hat{\lambda}_n^{(i)}\right)\ge 0 
\end{equation}
for $j=1,2, \ldots,i $ and $j = K$.
From \eqref{C.6.1}, it is found that 
\begin{equation}
\label{nvsK}
\hat{\lambda}^{(i)}_{n}=\hat{\lambda}^{(i)}_{K}\frac{\sqrt{w_{n}}}{\sqrt{w_{K}}} \quad \quad \mathrm{for} \;n=i+1,i+2,\ldots,K.
\end{equation}
Moreover, it is easily seen that $\varsigma_K^{(i)}>0$ from which using \eqref{C.7} it follows that the last constraint is always satisfied with strict equality, i.e., $\sum\nolimits_{n=1}^{K}\hat{\lambda}_n^{(i)} = \delta_{K}$.

Similarly, the solutions of $\mathcal P^{(i+1)}$ take the form
\begin{equation}\label{C.6.2}
\; \hat {\lambda}_n^{(i+1)} =    \sqrt{\frac{{w_n}}{{\sum\nolimits_{j=n}^{i+1} \varsigma_j^{(i+1)}+\varsigma_K^{(i+1)}}}}\quad \; {\rm{for}}\; n=1,2,\ldots,i+1 \quad\quad\;
\end{equation}
and
\begin{equation}\label{C.6.3}
\hat {\lambda}_n^{(i+1)} =  \sqrt{\frac{{w_n}}{{\varsigma_K^{(i+1)}}}}\quad \quad {\rm{for}}\; n=i+2,i+3,\ldots,K
\end{equation}
where the Lagrange multipliers $\varsigma_j^{(i+1)}$ are such that
\begin{equation}\label{C.7.1}
0 \le \varsigma_j^{(i+1)} \bot \left(\delta _j - \sum\nolimits_{n=1}^{j}\hat{\lambda}_n^{(i+1)}\right)\ge 0
\end{equation}
for $j=1,2, \ldots,i +1$ and $j = K$.

To proceed further, we focus on $\mathcal P^{(i+1)}$ and consider the two conditions $\varsigma_{i+1}^{(i+1)}>0$ and $\varsigma_{i+1}^{(i+1)}=0$, separately. If $\varsigma_{i+1}^{(i+1)}>0$, from \eqref{C.7.1} it follows that the $(i+1)$th constraint of $\mathcal P^{(i+1)}$ is satisfied with equality, i.e., 
\begin{equation}\label{C.8.1}
\sum\nolimits_{n=1}^{i+1}\hat \lambda_n^{(i)} = \delta _{i+1}.
\end{equation}
Also, since $\sum\nolimits_{n=1}^{K}\hat{\lambda}_n^{(i+1)} = \delta_{K}$ we may write
\begin{equation}
\label{ }
\sum\nolimits_{n=i+2}^{K}\hat{\lambda}_n^{(i+1)}=\sum\nolimits_{n=1}^{K}\hat{\lambda}_n^{(i+1)} -\sum\nolimits_{n=1}^{i+1}\hat{\lambda}_n^{(i+1)}=\delta_{K}-\delta_{i+1}.
\end{equation}
Recalling \eqref{nvsK} yields $\hat{\lambda}^{(i+1)}_{n}=\hat{\lambda}^{(i+1)}_{K}\frac{\sqrt{w_{n}}}{\sqrt{w_{K}}}$ for $n=i+2,\ldots,K$ from which we may write
\begin{equation}
\label{C.eq}
\hat{\lambda}_K^{(i+1)}=\sqrt{w_{K}}\frac{\delta_{K}-\delta_{i+1}}{\sum\nolimits_{n=i+2}^{K}\sqrt{w_{n}}}.
\end{equation}
In addition, it is easily recognized that if $\sum\nolimits_{n=1}^{i+1}\hat \lambda_n^{(i+1)} = \delta _{i+1}$ then 
\begin{equation}\nonumber
\sum\nolimits_{n=1}^{i+1} \hat \lambda_n^{(i)} \ge \delta _{i+1}
\end{equation}
so that
\begin{equation}\nonumber
\sum\nolimits_{n=i+2}^{K} \hat \lambda_n^{(i)} \le \delta_{K}- \delta _{i+1}.
\end{equation}
Substituting \eqref{nvsK} into the above inequality produces
\begin{equation}\label{C.gt}
\hat \lambda_K^{(i)} \le \sqrt{w_{K}}\frac{\delta_{K}-\delta_{i+1}}{\sum\nolimits_{n=i+2}^{K}\sqrt{w_{n}}}. 
\end{equation}
Putting \eqref{C.eq} and \eqref{C.gt}, it is found that when $\varsigma_{i+1}^{(i+1)}>0$ then
\begin{equation}\nonumber
\hat{\lambda}^{(i+1)}_{K}=\max\left(\hat{\lambda}^{(i)}_{K}, \sqrt{w_{K}}\frac{\delta_{K}-\delta_{i+1}}{\sum\nolimits_{n=i+2}^{K}\sqrt{w_{n}}}\right).
\end{equation}
On the other hand, if $\varsigma_{i+1}^{(i+1)}=0$ from (\ref{C.7.1}) we have 
\begin{equation}\nonumber
\sum\nolimits_{n=1}^{i+1}\hat \lambda_n^{(i+1)} < \delta _{i+1}. 
\end{equation}
Accordingly, we may write
\begin{equation}
\label{ }\nonumber
\sum\nolimits_{n=i+2}^{K}\hat \lambda_n^{(i+1)} > \delta_{K}-\delta _{i+1}
\end{equation}
and
\begin{equation}\label{C.12.1}
\hat \lambda_K^{(i+1)} > \sqrt{w_{K}}\frac{\delta_{K}-\delta_{i+1}}{\sum\nolimits_{n=i+2}^{K}\sqrt{w_{n}}}.
\end{equation}
In this case, it is easily seen from \eqref{C.6} -- \eqref{C.7} and \eqref{C.6.2} -- \eqref{C.7.1} that $\mathcal{P}^{(i)}$ and $\mathcal{P}^{(i+1)}$ have the same solution. In particular, this means that $\hat{\lambda}_K^{(i)} = \hat{\lambda}_K^{(i+1)}$. 

Collecting the above results togheter, it turns out that 
\begin{equation}\label{i+1vsi}
\hat{\lambda}^{(i+1)}_{K}=\max\left(\hat{\lambda}^{(i)}_{K}, \sqrt{w_{K}}\frac{\delta_{K}-\delta_{i+1}}{\sum\nolimits_{n=i+2}^{K}\sqrt{w_{n}}}\right)
\end{equation}
from which it follows that the solutions $\hat{\lambda}^{(i)}_{K}$ for any $i \le K-1$ can be easily computed once $\hat \lambda_K^{(0)}$ is given. The latter is easily obtained from \eqref{i+1vsi} and reads
\begin{equation}
\label{ }\nonumber
\hat{\lambda}_n^{(0)}=\sqrt{w_{n}}\frac{\delta_{K}}{\sum\nolimits_{j=1}^{K}\sqrt{w_{j}}} \quad \quad \mathrm{for} \quad n=1,2,\ldots,K
\end{equation}
where we have defined $\delta_0 \triangleq 0$. Using the above result and applying repeatedly \eqref{i+1vsi}, it turns out that the solution of interest given by $\hat {\lambda}_{K} = \hat{\lambda}_{K}^{(K-1)}$ can be determined in closed-form as follows
\begin{equation}
\label{lambdaK}
\hat {\lambda}_{K}=\sqrt{w_{K}} \underset{j=0,\ldots,K-1}{\max}\frac{\delta_{K}-\delta_{j}}{\sum\nolimits_{n=j+1}^{K}\sqrt{w_{n}}}.
\end{equation}

Once $\hat {\lambda}_{K}$ has been computed through \eqref{lambdaK}, ${\lambda}^\star_{K}$ can be determined by applying the following lemma.

\begin{lemma} The solution ${\lambda}^\star_{K}$ is always such that 
\begin{equation}\label{U.1}
{\lambda}^\star_{K}=\min(1,\hat{\lambda}_{K}).
\end{equation}
\end{lemma}
\begin{proof}
Using the KKT optimality conditions of (\ref{3.44}) it is seen that
\begin{equation}\label{U.2}
{\lambda}^\star_{n} =    \sqrt{\frac{{w_n}}{{\sum\nolimits_{j=n}^{K} \zeta_j} + \gamma_n}}\quad \; {\rm{for}}\; n=1,2,\ldots,K \quad\quad\;
\end{equation}
where the Lagrange multipliers $\{\zeta_j\}$ and $\{\gamma_n\}$ must be chosen to satisfy the following constraints:
\begin{equation}\label{U.4}
0 \le \zeta_j \, \bot \left(\delta_j - \sum\nolimits_{n=1}^{j}{\lambda}^\star_{n}\right)\ge 0 \quad j=1,2,\ldots,K
\end{equation}
and
\begin{equation}\label{U.5}
0 \le \gamma_n \, \bot \left({\lambda}^\star_n - 1\right)\ge 0 \quad n=1,2,\ldots,K.
\end{equation}
In order to prove \eqref{U.1}, the two cases ${\lambda}^\star_K < 1$ and ${\lambda}^\star_K = 1$ are considered separately. Assume first ${\lambda}^\star_K < 1$. Then, from  \eqref{U.5} we have $\gamma_K = 0$. In addition, since the solution of (\ref{3.44}) is such that ${\lambda}^\star_n \le {\lambda}^\star_{n+1}$ for $n=1,2,\ldots,K-1$, it follows that $\lambda^{\star}_{n}<1$ and $\gamma_n = 0$ $\forall n$. Accordingly,  
\eqref{U.2} -- \eqref{U.4} become formally equivalent to the KKT conditions of problem \eqref{C.2} (see \eqref{C.6} -- \eqref{C.7}  with $i=K-1$) meaning that the two problems have the same solution. In particular, ${\lambda}^\star_K =\hat{\lambda}_K=\min(1,\hat{\lambda}_K)$. Assume now ${\lambda}^\star_K = 1$. In this case, it is easily recognized that $\hat{\lambda}_K \ge 1$. Indeed, if by absurd $\hat{\lambda}_{K} < 1$ then the solution of \eqref{C.2} would coincide with the solution of (\ref{3.44}), and ${\lambda}^\star_K$ would be equal to $\hat{\lambda}_{K} $ (less than 1), which contradicts the hypothesis ${\lambda}^\star_K = 1$. Then, even in this case ${\lambda}^\star_K =1=\min(1,\hat{\lambda}_K)$ and the result in \eqref{U.1} easily follows.
\end{proof}

To see how all the above results can be used to find the solution of (\ref{3.44}), observe that when ${\lambda}_K^{\star}$ has been computed through \eqref{lambdaK} and \eqref{U.1}, the remaining solutions $\lambda_n^\star$ for $n=1,2,\ldots,K-1$ can be found by solving the following problem
\begin{equation}\label{C.27.000}
    \begin{array}{clc}
      \underset{\{0 <\lambda_n \, \le \, 1\}}{\min} &
\sum\nolimits_{n=1}^{K-1}\dfrac{w_n}{\lambda_n}
&  \\
      \mathrm{s.t.} & \sum\nolimits_{n=1}^{j}\lambda_n \le \delta_{j} & j=1,2,\ldots,K-1 \\ 
      & \sum\nolimits_{n=1}^{K-1}\lambda_n \le \delta_{K} -{\lambda}_K^{\star}.
    \end{array}
\end{equation}
Recalling that $\delta _j = \sum \nolimits_{n=1}^{j} \eta_n$ we may write
\begin{equation}\nonumber
   \delta_{K} - {\lambda}_K^{\star} = \sum \nolimits_{n=1}^{K} \eta_n - {\lambda}_K^{\star}.
   \end{equation}
Assume now ${\lambda}_K^{\star} = 1$. Observing that $0< \eta_n \le 1$ we have 
\begin{equation}\nonumber
   \delta_{K} - {\lambda}_K^{\star} \le \delta_{K-1}.
   \end{equation}
The same result can be obtained when ${\lambda}_K^{\star} =\hat \lambda_K < 1$ using \eqref{lambdaK}. This means that \eqref{C.27.000} reduces to:
\begin{equation}\label{red1}
    \begin{array}{clc}
      \underset{\{0 <\lambda_n \, \le \, 1\}}{\min} &
\sum\nolimits_{n=1}^{K-1}\dfrac{w_n}{\lambda_n}
&  \\
      \mathrm{s.t.} & \sum\nolimits_{n=1}^{j}\lambda_n \le \delta^{\prime}_{j} & j=1,\ldots,K-1
    \end{array}
\end{equation}
where 
\begin{equation}\nonumber
\delta^{\prime}_{j} = \left\{ {\begin{array}{*{20}c}
   \delta_{j} & \quad {\rm{for}}\;j=1,2,\ldots,K-2  \\ 
   \delta_{K} -\lambda^{\star}_{K} & \;{\rm{for}}\; j=K-1. \quad \quad \quad\quad\\
\end{array}} \right.
\end{equation}
%
%

We now proceed computing $\lambda_{K-1}^\star$. For this purpose, we follow the same procedure used for the computation of $\lambda_{K}^\star$. We first compute $\hat \lambda_{K-1}$, which is the largest solution of the optimization problem obtained from \eqref{red1} after removing the constraints $\lambda_{n} \le 1$ for $n=1,2,\ldots, K-1$. Paralleling the same steps leading to \eqref{lambdaK} yields 
%
\begin{equation}
\label{}\nonumber
\hat{\lambda}_{K-1}=\sqrt{w_{K-1}} \cdot \underset{j=0,\ldots,K-2}{\max}\frac{\delta^{\prime}_{K-1}-\delta^{\prime}_{j}}{\sum\nolimits_{l=j+1}^{K-1}\sqrt{w_{l}}}
\end{equation}
from which $\lambda_{K-1}^\star$ is obtained as 
\begin{equation}\nonumber
\lambda_{K-1}^\star = \min{(1,\hat{\lambda}_{K-1})}.
\end{equation}
%
Applying repeatedly the same steps for any $n\le K-2$ leads to the iterative procedure illustrated in {\bf{Algorithm 1}} in the text.

\section*{Appendix D}


The convexity of the objective function $Q_n({\theta_{n}})$ in (\ref{func_Q}) on $\mathcal Q = \{\theta_{n}; \theta_n \le 0\}$ is studied. 
To simplify the notation, we omit the index $n$.
Taking the first derivative of $Q\left(\theta\right)$ produces
\begin{equation}\label{F.4}
Q{'}\left(\theta\right) =-\frac{\rho}{\sqrt{\lambda_{H_1}\lambda_{H_2}}} \left(\gamma + \frac{2 - e^{\theta}}{\sqrt{1-e^{\theta}}}\right)e^{-\theta}
\end{equation} 
from which it follows that 
\begin{equation}\nonumber
Q{''}\left(\theta\right) =\frac{\rho}{\sqrt{\lambda_{H_1}\lambda_{H_2}}} \left[\gamma + \frac{e^{2\theta}-6 e^{\theta} +4}{2\big(1-e^{\theta}\big)^{3/2}}\right]e^{-\theta}.
\end{equation} 
Setting $Q{''}\left(\theta\right) = 0$ and letting $\phi=e^{\theta}$ yields $\phi^{2}-6\phi+4= -2\gamma(1-\phi)^{3/2}$
or, equivalently, taking the square of both sides
\begin{equation}\nonumber
\phi^4 - 4\left(3-\gamma^{2}\right)\phi^3+4\left(11-3\gamma^{2}\right)\phi^2 + 4\left(4-\gamma^{2}\right)\left(1-3\phi\right) = 0.
\end{equation}
It turns out that the above equation has a unique solution in the interval $(0,1]$, which is found to be in the form of \eqref{4.18}. As seen, $\phi$ in \eqref{4.18} results to be a monotonic increasing function of $\gamma$. Recalling that $\gamma$ takes values in the interval $[2,\infty)$, we have that the minimum value $\phi_{{\rm{min}}}$ is achieved when $\gamma=2$. Setting $\gamma= 2$ into the above quartic equation produces $\phi^2(\phi^2 + 4\phi-4) = 0$ from which \eqref{4.18.1} is easily found.

On the basis of the above results, it follows that $Q{''}\left(\theta\right)$ is zero in $\mathcal Q$ only for $\theta=\ln \phi$. Moreover, it turns out that $Q{''}\left(\theta\right) > 0$ for $\theta \le \ln \phi$ and $Q{''}\left(\theta\right) < 0$ for $\ln \phi < \theta \le 0$. Then, we may conclude that $Q\left(\theta\right)$ is convex for $\theta \le \ln \phi$ and concave for $\ln \phi < \theta \le 0$.

The point $\psi$ in \eqref{4.29} belongs to the interval $(0,1]$ and it is such that 
\begin{equation}\nonumber
Q(\ln \psi) + Q'(\ln \psi)\left(\theta - \ln \psi\right) = Q'(\ln \psi)\theta
\end{equation}
for $\ln \psi < \theta \le 0$. This amounts to saying that $Q'(\ln \psi)\ln \psi = Q(\ln \psi)$
from which using (\ref{func_Q}) and \eqref{F.4} we obtain
\begin{equation}\nonumber
\gamma_{n} \left(\frac{1}{{\psi}} - 1\right) +  2 \frac{\sqrt{1- {\psi}}}{{\psi}} = -\frac{\ln \psi}{\psi}\left(\gamma + \frac{2 - {\psi}}{\sqrt{1-{\psi}}}\right).
\end{equation}
The above equation can be easily rewritten as in \eqref{4.32} in the text. 

\section*{Appendix E }
The proof of Proposition 2 is divided into the same steps used in Appendix C  for Proposition 1. In the sequel, we report only the major differences and refer to Appendix C for the complete proof. The first step removes the constraints $\theta_n \le 0$ and proceeds computing $\hat {\theta}_{K}$ from $ {\theta}_{K}^\star$ is then obtained using the following lemma.

\begin{lemma} The solution ${\theta}^\star_{K}$ is such that
\begin{equation}\label{Z.1}
{\theta}^\star_{K}=\min(0,\hat {\theta}_{K})
\end{equation}
where $\hat {\theta}_{K}$ is given by 

\begin{equation}\label{H.7.1}
\hat \theta_K= \ln w_K + 
\mathop {\max }\limits_{\ell=0,1,\ldots,K-1} \dfrac{ \delta_K - \delta_{\ell}- \sum\nolimits_{j=\ell+1}^K \ln w_j}{K - \ell}.
%
\end{equation}
with 
\begin{equation}
\delta_{j}= \sum\limits_{n=1}^{j}\kappa_n = \sum\limits_{n=1}^j \ln \eta_n.
\end{equation}
for $j=1,2,\ldots,K$ and $\delta_0 \triangleq 0$.
\end{lemma}

\begin{proof} The proof follows using the same arguments adopted in Appendix C and it is not reported for space limitations. 
\end{proof}

Once $\theta_K^\star$ is obtained using \eqref{Z.1} and \eqref{H.7.1}, we proceed computing $\theta_{K-1}^\star$. Paralleling the same steps for $\theta_K^\star$, the latter is found to be
\begin{equation}\label{pippo}
{\theta}^\star_{K-1}=\min(0,\hat {\theta}_{K-1})
\end{equation}
where $\hat {\theta}_{K-1}$ is obtained solving
\begin{equation}\nonumber
    \begin{array}{cll}
      \;\quad \quad \quad \quad \quad\underset{\{\theta_n\}}{\min} &
\sum\nolimits_{n=1}^{K-1} w_n e^{-\theta_n}
&  \\
      \;\;\quad\quad\quad\quad\quad \mathrm{s.t.} & \sum\nolimits_{n=1}^{j} \theta_n \le \delta _j & \!\!\!\!\!\!\!\!\!\!\!\! j=1,2,\ldots,K-1 \\
            \;\;\quad\quad\quad\quad\quad \; & \sum\nolimits_{n=1}^{K-1} \theta_n \le \delta _K - \theta_K^\star.
    \end{array}
\end{equation}
{Observing that $\delta_{K} - \theta_K^\star$ is always such that
\begin{equation}\nonumber
   \delta_{K} - \theta_K^\star \le \delta_{K-1}
   \end{equation}
the above problem reduces to the following one:
\begin{equation}\nonumber
    \begin{array}{clc}
      \underset{\{\theta_n\}}{\min} &
\sum\nolimits_{n=1}^{K-1}{w_n}e^{-\theta_n}
&  \\
      \mathrm{s.t.} & \sum\nolimits_{n=1}^{j}\theta_n \le \delta^{\prime}_{j} & j=1,\ldots,K-1
    \end{array}
\end{equation}
where 
\begin{equation}\nonumber
\delta^{\prime}_{j} = \left\{ {\begin{array}{*{20}c}
   \delta_{j} & \quad {\rm{for}}\;j=1,2,\ldots,K-2  \\ 
   \delta_{K} - \theta_K^\star& \;{\rm{for}}\; j=K-1. \quad \quad \quad\quad\\
\end{array}} \right.
\end{equation}
%
%
The solution $\hat \theta_{K-1}$ in \eqref{pippo} has the same form of \eqref{H.7.1} once $K$ is replaced with $K-1$ and the quantities $\delta_{\ell}$ are replaced with $\delta^{\prime}_{\ell}$ for $\ell =1,2,\ldots,K-1$. Applying repeatedly the same steps for any $n\le K-2$ leads to the iterative procedure illustrated in {\bf{Algorithm 2}} in the text.}

\bibliographystyle{IEEEtran}
\bibliography{IEEEabrv,bibnew}

\begin{thebibliography}{10}
\providecommand{\url}[1]{#1}
\csname url@samestyle\endcsname
\providecommand{\newblock}{\relax}
\providecommand{\bibinfo}[2]{#2}
\providecommand{\BIBentrySTDinterwordspacing}{\spaceskip=0pt\relax}
\providecommand{\BIBentryALTinterwordstretchfactor}{4}
\providecommand{\BIBentryALTinterwordspacing}{\spaceskip=\fontdimen2\font plus
\BIBentryALTinterwordstretchfactor\fontdimen3\font minus
  \fontdimen4\font\relax}
\providecommand{\BIBforeignlanguage}[2]{{%
\expandafter\ifx\csname l@#1\endcsname\relax
\typeout{** WARNING: IEEEtran.bst: No hyphenation pattern has been}%
\typeout{** loaded for the language `#1'. Using the pattern for}%
\typeout{** the default language instead.}%
\else
\language=\csname l@#1\endcsname
\fi
#2}}
\providecommand{\BIBdecl}{\relax}
\BIBdecl

\bibitem{Rong2011}
Y.~Rong, ``Multihop non-regenerative {MIMO} relays: {QoS} considerations,''
  \emph{IEEE Trans. Signal Process.}, vol.~59, no.~1, pp. 290 -- 303, 2011.

\bibitem{Tarokh1999}
V.~Tarokh, H.~Jafarkhani, and A.~Calderbank, ``Space-time block codes from
  orthogonal designs,'' \emph{IEEE Trans. Inf. Theory}, vol.~45, no.~5, pp.
  1456 --1467, July 1999.

\bibitem{Telatar1999}
I.~Telatar, ``Capacity of multi-antenna {Gaussian} channels,'' \emph{European
  Transactions on Telecommunications}, vol.~10, no.~10, pp. 585 -- 595,
  Nov./Dec. 1999.

\bibitem{Foschini98onlimits}
G.~J. Foschini and M.~J. Gans, ``On limits of wireless communications in a
  fading environment when using multiple antennas,'' \emph{Wireless Personal
  Communications}, vol.~6, pp. 311-- 335, 1998.

\bibitem{PaulrajBook}
A.~Paulraj, R.~Nabar, and D.~Gore, \emph{Introduction to Space-Time Wireless
  Communications}.\hskip 1em plus 0.5em minus 0.4em\relax Cambridge, U.K.:
  Cambrdige University Press, 2003.

\bibitem{PalomarBook}
D.~P. Palomar and Y.~Jiang, \emph{{MIMO} Transceiver Design via Majorization
  Theory}.\hskip 1em plus 0.5em minus 0.4em\relax Now Publishers, vol. 3, no. 4
  -- 5, pp. 331-551: Foundations and Trends in Communications and Information
  Theory, 2006.

\bibitem{Shenouda2008}
M.~Shenouda and T.~Davidson, ``A framework for designing {MIMO} systems with
  decision feedback equalization or {Tomlinson-Harashima} precoding,''
  \emph{IEEE J. Sel. Areas Commun.}, vol.~26, no.~2, pp. 401 -- 411, Feb. 2008.

\bibitem{D'Amico2008}
A.~A. D'Amico, ``Tomlinson-{H}arashima precoding in {MIMO} systems: a unified
  approach to transceiver optimization based on multiplicative
  {S}chur-convexity,'' \emph{IEEE Trans. Signal Process.}, vol.~56, no.~8, pp.
  3662 -- 3677, Aug. 2008.

\bibitem{MarshallBook}
A.~W. Marshall and I.~Olkin, \emph{Inequalities:Theory of majorization},
  A.~press, Ed.\hskip 1em plus 0.5em minus 0.4em\relax Academic press, 1979.

\bibitem{SanguinettiJSAC2012}
L.~Sanguinetti, A.~A. D'Amico, and Y.~Rong, ``A tutorial on the optimization of
  amplify-and-forward {MIMO} relay networks,'' \emph{{\rm{submitted to}} IEEE
  J. Sel. Areas Commun. (Special issue on Theories and Methods for Advanced
  Wireless Relays)}, Aug. 2011, {\rm{available on line at
  http://www.iet.unipi.it/l.sanguinetti/}}.

\bibitem{Tang2007}
X.~Tang and Y.~Hua, ``Optimal design of non-regenerative {MIMO} wireless
  relays,'' \emph{IEEE Trans. Signal Process.}, vol.~6, no.~4, pp. 1398 --
  1407, April 2007.

\bibitem{Medina2007}
O.~Munoz-Medina, J.~Vidal, and A.~Agustin, ``Linear transceiver design in
  non-regenerative relays with channel state information,'' \emph{IEEE Trans.
  Signal Process.}, vol.~55, no.~6, pp. 2593 -- 2604, June 2007.

\bibitem{Rong2009LinearRelayCommunication}
Y.~Rong, X.~Tang, and Y.~Hua, ``A unified framework for optimizing linear
  non-regenerative multicarrier {MIMO} relay communication systems,''
  \emph{IEEE Trans. Signal Process.}, vol.~57, no.~12, pp. 4837 -- 4851, Dec.
  2009.

\bibitem{Rong2009LinearMultiHop}
Y.~Rong and Y.~Hua, ``Optimality of diagonalization of multi-hop {MIMO}
  relays,'' \emph{IEEE Trans. Wireless Commun.}, vol.~8, no.~12, pp. 6068 --
  6077, 2009.

\bibitem{Palomar03}
D.~P. Palomar, J.~M. Cioffi, and M.~A. Lagunas, ``Joint {Tx-Rx} beamforming
  design for multicarrier {MIMO} channels: a unified framework for convex
  optimization,'' \emph{IEEE Trans. Signal Process.}, vol.~51, pp. 2381 --
  2401, 2003.

\bibitem{Rong2009NonLinearRelayCommunication}
Y.~Rong, ``Optimal linear non-regenerative multi-hop {MIMO} relays with
  {MMSE-DFE} receiver at the destination,'' \emph{IEEE Trans. Wireless
  Commun.}, vol.~9, no.~7, pp. 2268 -- 2279, July 2010.

\bibitem{Zhang2005}
J.-K. Zhang, A.~Kavcic, and K.~M. Wong, ``Equal-diagonal {QR} decomposition and
  its application to precoder design for successive-cancellation detection,''
  \emph{IEEE Trans. Inf. Theory}, vol.~51, no.~1, pp. 154 -- 172, Jan. 2005.

\bibitem{Sampath01generalizedlinear}
H.~Sampath, P.~Stoica, and A.~Paulraj, ``Generalized linear precoder and
  decoder design for {MIMO} channels using the weighted {MMSE} criterion,''
  \emph{IEEE Trans. Commun.}, vol.~49, pp. 2198 -- 2206, 2001.

\bibitem{PalomarQoS2004}
D.~P. Palomar, M.~Lagunas, and J.~Cioffi, ``Optimum linear joint
  transmit-receive processing for {MIMO} channels with {QoS} constraints,''
  \emph{IEEE Trans. Signal Process.}, vol.~52, no.~5, pp. 1179 -- 1197, May
  2004.

\bibitem{Jiang2006}
Y.~Jiang, W.~Hager, and J.~Li, ``Tunable channel decomposition for {MIMO}
  communications using channel state information,'' \emph{IEEE Trans. Signal
  Process.}, vol.~54, no.~11, pp. 4405 --4418, 2006.

\bibitem{Guan2008QoSConstraints}
W.~Guan, H.~Luo, and W.~Chen, ``Linear relaying scheme for {MIMO} relay system
  with {QoS} requirements,'' \emph{IEEE Signal Processing Letters}, vol.~15,
  pp. 697 -- 700, 2008.

\bibitem{Sayed2007QoSConstraints}
N.~Khajehnouri and A.~Sayed, ``Distributed {MMSE} relay strategies for wireless
  sensor networks,'' \emph{IEEE Trans. Signal Process.}, vol.~55, no.~7, pp.
  3336 -- 3348, July 2007.

\bibitem{Mohammadi2010}
J.~Mohammadi, F.~Gao, and Y.~Rong, ``Design of amplify and forward {MIMO} relay
  networks with {QoS} constraint,'' in \emph{Proceedings of the IEEE Global
  Telecommunications Conference}, Miami, FL, USA, Dec. 2010, pp. 1 -- 5.

\bibitem{Viswanath1999}
P.~Viswanath and V.~Anantharam, ``Optimal sequences and sum capacity of
  synchronous {CDMA} systems,'' \emph{IEEE Trans. Inf. Theory}, vol.~45, no.~6,
  pp. 1984 -- 1991, Sep. 1999.

\bibitem{BoydBook}
S.~Boyd and L.~Vandenberghe, \emph{Convex optimization}, C.~U. Press, Ed.\hskip
  1em plus 0.5em minus 0.4em\relax Cambridge University Press, 2002.

\bibitem{Gao2008}
F.~Gao, T.~Cui, and A.~Nallanathan, ``Optimal training design for channel
  estimation in decode-and-forward relay networks with individual and total
  power constraints,'' \emph{IEEE Trans. Signal Process.}, vol.~56, no.~12, pp.
  5937 -- 5949, Dec. 2008.

\bibitem{Fu2011}
Y.~Fu, L.~Yang, W.-P. Zhu, and C.~Liu, ``Optimum linear design of two-hop
  {MIMO} relay networks with {QoS} requirements,'' \emph{IEEE Trans. Signal
  Process.}, vol.~59, no.~5, pp. 2257 -- 2269, May 2011.

\bibitem{Jiang06thegeneralized}
Y.~Jiang, W.~W. Hager, and J.~Li, ``The generalized triangular decomposition,''
  \emph{Math. Comp}, vol.~77, pp. 1037 -- 1056, 2008.

\end{thebibliography}

\end{document}